\documentclass[12pt, draftclsnofoot, onecolumn]{IEEEtran}
\pdfoutput=1
\usepackage{sharina}
\usepackage{setspace}
\usepackage{amssymb,amsmath,latexsym}
\usepackage{graphicx,subfigure}
\usepackage{mathrsfs}
\usepackage{amsmath}
\usepackage{color,cite,times}
\usepackage{epstopdf}
\usepackage{multirow}
\usepackage{array}
\newtheorem{assumption}{Assumption}

\newcommand{\slcurly}[1]{$\left\{\vcenter{\hbox{$\shortstack{#1}$}}\right.$}

\newtheorem{lemma}{\textbf{Lemma}}

\newtheorem{proposition}{\textbf{Proposition}}
\theoremstyle{remark}
\newtheorem{remark}{\textbf{Remark}}
\def\IEEEproofname{Proof}
\allowdisplaybreaks[4]
\begin{document}
	\title{Decentralized Federated Learning via MIMO Over-the-Air Computation: Consensus Analysis and Performance Optimization}
	\author{
		Zhiyuan~Zhai, 
		Xiaojun~Yuan,~\IEEEmembership{Senior Member, IEEE}, 
		and~Xin~Wang,~\IEEEmembership{Fellow, IEEE}
	}
	
	\maketitle
	
	\vspace{-2.6em}
	\begin{abstract}
		\vspace{-.3em}
		Decentralized federated learning (DFL), inherited from distributed optimization, is an emerging paradigm to leverage the explosively growing data from wireless devices in a fully distributed manner. With the cooperation of edge devices, DFL enables joint training of machine learning model under device to device (D2D) communication fashion without the coordination of a parameter server. However, the deployment of wireless DFL is facing some pivotal challenges. Communication is a critical bottleneck due to the required extensive message exchange between neighbor devices to share the learned model. Besides, consensus becomes increasingly difficult as the number of devices grows because there is no available central server to perform coordination. To overcome these difficulties, this paper proposes employing over-the-air computation (Aircomp) to improve communication efficiency by exploiting the superposition property of analog waveform in multi-access channels, and introduce the mixing matrix mechanism to promote consensus using the spectral property of symmetric doubly stochastic matrix. Specifically, we develop a novel multiple-input multiple-output over-the-air DFL (MIMO  OA-DFL) framework to study over-the-air DFL problem over MIMO multiple access channels. We conduct a general convergence analysis to quantitatively capture the influence of aggregation weight and communication error on the MIMO  OA-DFL performance in  \emph{ad hoc} networks. The result shows that the communication error together with the spectral gap of mixing matrix has a significant impact on the learning performance. Based on this, a joint communication-learning optimization problem is formulated to optimize transceiver beamformers and mixing matrix. Extensive numerical experiments are performed to reveal the characteristics of different topologies and demonstrate the substantial learning performance enhancement of our proposed algorithm.
	\end{abstract}


	\section{Introduction}	 
	Empowered by an unprecedented increase in local data generated by mobile edge devices, there is a surging trend in developing deep learning applications at the edge of wireless networks. These applications encompass various domains, including image recognition \cite{he2016deep} and natural language processing \cite{young2018recent}.
	However, primarily due to the requirement of collecting distributed data for centralized training, traditional machine learning (ML) approaches face limitations in terms of communication bandwidth and potential privacy concerns. 
	Federated learning (FL) is a distributed machine learning paradigm that has the ability to address these challenges \cite{konevcny2016federated}. 
	FL enables participating mobile devices to train a global learning model with the coordination of a parameter server (PS). In this approach, each device computes local model updates, such as model parameters or gradients, by utilizing its local datasets. These updates are then uploaded to the PS, where the averaged model is computed and subsequently broadcasted to the devices.
	
	One significant limitation of FL is its heavy reliance on the central PS. FL requires aggregating all device updates at the PS, resulting in communication congestion and reduced fault tolerance. This bottleneck makes it challenging for FL to handle a massive number of devices efficiently. 
	Moreover, in certain application scenarios like autonomous robotics and collaborative driving\cite{savazzi2021opportunities}, centralized FL may not be reliable due to the absence of an available central PS.
	To address these drawbacks, decentralized federated learning (DFL) has emerged as a promising alternative. DFL eliminates the need for coordination from a central PS by enabling each device to maintain and optimize its local model. Model exchange is achieved through device-to-device (D2D) communications. 
	The concept of decentralized learning/optimization traces back to the 1980s\cite{tsitsiklis1984problems}, with algorithms like the alternating direction method of multipliers (ADMM)\cite{wei2012distributed}, dual averaging \cite{duchi2011dual}, and gradient descent \cite{nedic2009distributed} being well-known in this field.
	More recently, decentralized stochastic gradient descent (DSGD) \cite{ram2008distributed}, \cite{lian2017can} has gained attention as a novel algorithm for large-scale deep learning problems. DSGD ensures convergence to optimality under the assumptions on convexity, gradient, and network connectivity. This framework has been extended to accommodate various network paradigms and enhance convergence rates. For example, in \cite{basu2019qsparse}, the authors propose a scheme involving joint quantization, aggressive sparsification, and local computations to alleviate communication overhead. Additionally, \cite{koloskova2020unified} presents a comprehensive convergence analysis that encompasses local SGD updates, synchronous updates, and pairwise gossip processes on changing topologies.

	Despite the promising potential of DFL, most of the existing works suppose error free communication links between devices while the real-world communication systems are prone to distortions.  Imperfect communication conditions, including limited wireless resources, channel fading, noise, and mutual interference, can result in inaccurate model exchanges, thus hindering training performance.  Additionally, transmitting model parameters through D2D communications can introduce significant communication overhead, which limits the scalability of DFL \cite{ye2022decentralized}.  
	To tackle these challenges, several recent works have focused
	on the communication aspect of DFL and proposed over-the-air computation (Aircomp) \cite{nazer2007computation} to improve
	the communication efficiency in the aggregation process.
	Aircomp leverages the superposition property of electromagnetic waves, enabling edge devices to transmit their model parameters simultaneously using shared radio resources. The signal is then aggregated in the wireless channel, allowing the receiver to obtain an approximation of the desired aggregated value.
	For instance, \cite{xing2020decentralized} uses a heuristic greedy coloring algorithm to arrange the communication order and enable devices to perform computational over-the-air sequentially in successive slots under D2D networks. Similarly, \cite{shi2021over} separates the communication process into scheduling and transmission parts and schedules the selected device as the active central server to enable interference-free over-the-air transmission. The authors in \cite{lin2022distributed} propose a one-step over-the-air scheme where all devices exchange model parameters in a single phase via full-duplex (FD) communication to accelerate the training speed.
	
	Nevertheless, these recent works have their limitations. Particularly,  
	\cite{xing2020decentralized} and \cite{shi2021over} determine the mixing matrix based on standard examples, which may not be suitable for specific DFL systems or changing wireless conditions. 
	Moreover, the heuristic protocols employed in their system designs do not guarantee the optimality of DFL performance.
	Although \cite{lin2022distributed} has evidenced the effectiveness of the over-the-air technique in improving DFL model aggregation performance, their work only focuses on beamforming optimization in fully connected topology. Hence, the lack of consideration for learning aspects and various network topologies limits the full potential release of DFL systems. Therefore, there is a pressing need to conduct theoretical analysis and performance optimization to address general DFL scenarios from a joint communication-learning perspective.

	In this paper, we present a novel multiple-input multiple-output over-the-air decentralized federated learning (MIMO OA-DFL) scheme.  To fully harness the potential of wireless DFL performance, we develop a general communication-learning framework for the considered MIMO OA-DFL system. Furthermore, we conduct convergence analysis to characterize the impact of mixing matrix and communication error on the DFL learning accuracy under moderate assumptions. Based on this analysis, we propose a low-complexity algorithm that utilizes alternating optimization (AO) to jointly optimize the mixing matrix and transceiver beamformers. We summarize our contributions as follows.
	\begin{itemize}
		\item  We investigate the DFL problem in general \emph{ad hoc} networks and establish a joint communication and learning framework for the considered MIMO OA-DFL scheme. In this framework, we introduce mixing matrix mechanism to guarantee consensus together with beamforming design to improve communication quality.
		\item  We derive a rigorous convergence bound for the global loss function. This bound is obtained by utilizing the symmetric doubly stochastic character of mixing matrix and the statistical properties of communication errors. To the best of our knowledge, our derivation is the first analysis on the convergence of decentralized learning/optimization in the presence of communication error and is applicable to arbitrary topologies.
		Based on our convergence analysis, we formulate the communication (beamformers) and learning (mixing matrix) joint optimization problem to enhance MIMO OA-DFL performance.
		\item We propose an efficient AO algorithm \cite{corne1999new} to obtain the solution of transceiver beamformers and mixing matrix. Particularly, we transform the optimization of multicast beamforming into a convex quadratically constrained quadratic programming (QCQP) problem and determine the mixing matrix using monotonicity of the objective function and variational characterization of optimization variables\cite{boyd_vandenberghe_2004}.
	\end{itemize}
	
	Simulation results demonstrate the effectiveness of the proposed scheme and shed light on the characteristics of different topologies in MIMO OA-DFL. Specifically, our numerical results on the error-free case validate the precision of the derived convergence bound. 
	We also conduct an in-depth analysis of the trade-off between communication and learning by analyzing the performance differences among various topologies.
	Furthermore, the comparisons with benchmark methods show that our scheme achieves significant performance improvements and near-optimal learning accuracy.
	
	The remainder of this paper is organized as follows. In Section II, we provide details of the DFL learning and communication models. Section III introduces the proposed MIMO OA-DFL framework. Section IV presents the preliminary assumptions and analyzes the convergence of MIMO OA-DFL. In Section V, we formulate the performance optimization problem that minimizes the global training loss and propose algorithms to jointly optimize the beamformers and mixing matrix. Section VI presents the simulation results, and we conclude the paper with remarks in Section VII.

	\textit{Notations:}
	We use the set notation $[M]$ to denote the set $\{i|1\leq i \leq M\}$, and denote the real and complex number sets by $\mathbb{R}$ and $\mathbb{C}$, respectively. The regular letters, lowercase letters in bold, and bold capital letters are used to denote scalars, vectors and matrices, respectively. We use $(\cdot)^\ast$, $(\cdot)^\mathrm{T}$,  $(\cdot)^\mathrm{H}$, and $(\cdot)^{\dagger}$ to denote the conjugate, the transpose, the conjugate transpose, and the pseudoinverse, respectively. We use $x[i]$ to denote the $i$-th entry of vector $\xx$, $x_{ij}$ to denote the $(i,j)$-th entry of matrix $\mX$, $\mathcal{CN}(\mu,\sigma^2)$ to denote circularly-symmetric complex normal distribution with mean $\mu$ and covariance $\sigma^2$. The $l_2$-norm is denoted by $\norm{\cdot}$, while the Frobenius norm is denoted by $\norm{\cdot}_F$.  The expectation operator is represented by $\E$. We use $\1_n$ to denote the column vector in $\mathbb{R}^n$ with all elements being 1, and $\1$ to denote such a vector with the appropriate dimension. We use $\Tr(\cdot)$ to denote the trace of a square matrix. The identity matrix is denoted by $\mathbf{I}$, while $\lambda_i(\cdot)$ denotes the $i$-th largest eigenvalue of a matrix. We use $\nabla f(\cdot)$ to denote the gradient of a function $f$, and $\partial F(\cdot)$ to denote the concatenation of all gradients of the devices.
	
	\section{Learning and Communication Models} \label{section-system-model}
	In this section, we discuss the DFL process and present the underlying communication channel to support data exchanges involved in the DFL process.
				\vspace{-1em}
	\subsection{Decentralized Federated Learning}\label{sec2a}
	We begin with a description of the DFL system where $M$ devices cooperatively train a machine learning model. The common objective of the $M$ devices is to minimize an empirical loss function 
	\begin{align}
		\  f({\bf{x}} )=\frac{1}{M} \sum_{i=1}^M f_i({\bf{x}} ), \  \label{eq:f}
	\end{align} 
	where ${\bf{x}}\in {\mathbb{R}}^D$ is the  model parameter with dimension $D$,
	and $f_i \colon {\mathbb{R}}^D \to {\mathbb{R}}$  is the local loss function of device $i$ defined by
				\vspace{-1em}
	\begin{align}\label{eq:f_m}
		f_i({\bf{x}}) := {\mathbb{E}}_{\xi_i \sim {\mathcal D}_i} F({\bf{x}},\xi_i),
	\end{align}
	with ${\mathcal D}_i$ being the predefined distribution of local data samples on device $i$, and $F({\bf{x}},\xi_i)$ being the loss function with respect to samples $\xi_i$. 
	\begin{figure}[htbp]
		\centering
		\includegraphics[width=0.6\textwidth]{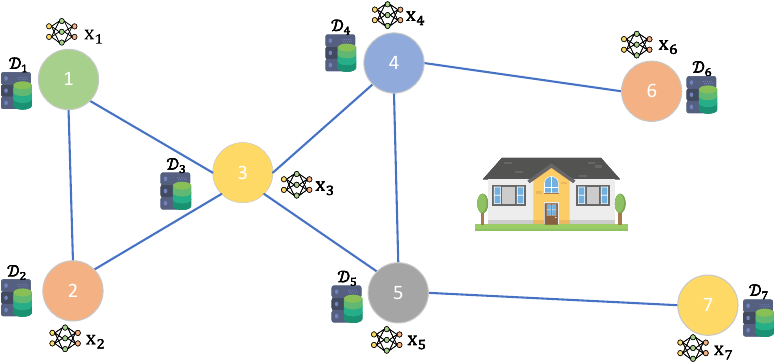}
		\caption{An example of the DFL system with seven devices.}
		\label{fig_system_model}
					\vspace{-1em}
	\end{figure}

	The devices update their local models  by minimizing their individual local loss functions, and then exchange learned model parameters via communication links
	to promote decentralized training. 
	Let $e_{ij}$ be an indicator function of the communication link between device $i$ and device $j$. That is, $e_{ij}=1$ if the communication link between device $i$ and device $j$ exists, and $e_{ij}=0$ otherwise. We assume full-duplex communication, i.e., $e_{ij}=e_{ji}$. Then, the
	communication topology for model exchanges can be represented by an undirected graph $\mathcal{G}=\left(\mathcal{V},\mathcal{E}\right)$, where $\mathcal{V}$ represents the device set and $\mathcal{E}$ represents the set of all communication links, i.e., $\mathcal{E}=\{e_{ij}| e_{ij}=1, \forall i,j\}$. 
	We say that device $i$ is a neighbor of device $j$ if $e_{ij}=1$. An example of $\mathcal{G}$ is shown in Fig.~\ref{fig_system_model}. We assume that the communication topology remains unchanged during the whole training process.

	We now describe the training procedure of the DFL system. Specifically, we adopt the stochastic gradient descent method\cite{zinkevich2010parallelized} for local training, where the model parameters of all devices are iteratively updated at each training round.  At the 
	$t$-th round, the training process consists of the following three steps:
	\begin{itemize}
		\item \textit{Local gradient computation}: Each device $i$
		computes the local stochastic gradient $ \nabla F(\xx_i^{(t)}, \xi_i^{(t)})$ by randomly sampling $\xi_i^{(t)}$ in local training dataset ${\mathcal D}_i$, where $\xx_i^{(t)}$ denotes the model parameter of device $i$ in round $t$.
		\item \textit{Gossip model aggregation}: Devices communicate with their neighbors to exchange model parameters. Each device fetches the model parameters from its neighbors through wireless channels. Based on the received signals, each device estimates the weighted average as
		\begin{align}\label{free_aggre}
			\xx_{i}^{(t + \frac{1}{2})} = \sum_{j=1}^{M} {w_{ij}}\xx_j^{(t)},~ \forall i \in [M]
		\end{align}
		where ${w_{ij}\in [0,1]}$ is  the weighting factor for device $j$ aggregating on device $i$. Note that $w_{ij}=0$ if device $i$ does not have a communication link with device $j$.  We refer to  
		$\xx_{i}^{(t + \frac{1}{2})}$ as the ideal (error-free) aggregation model at device $i$,  and denote by $\hat{\xx}_{i}^{(t+\frac{1}{2})}$ an estimate of $\xx_{i}^{(t + \frac{1}{2})}$. Due to the presence of communication noise and channel fading, the estimate $\hat{\xx}_{i}^{(t+\frac{1}{2})}$ generally contains distortion, i.e., $\hat{\xx}_{i}^{(t+\frac{1}{2})} \neq \xx_{i}^{(t + \frac{1}{2})},\forall i \in [M]$.
		\item \textit{Local model update}: Based on the estimate $\hat{\xx}_{i}^{(t + \frac{1}{2})},\forall i \in [M]$, each device updates the local model parameter as
					\vspace{-1em}
		\begin{align}\label{update}
			\xx_{i}^{(t+1)} = \hat{\xx}_{i}^{(t+\frac{1}{2})} - \lambda \nabla F(\xx_i^{(t)}, \xi_i^{(t)}),~\forall i \in [M],
		\end{align}
		where $\lambda \in \R$ represents the learning rate.
	\end{itemize}
	
	The weighting factor of all devices can be captured by a mixing matrix, also known as gossip matrix  \cite{koloskova2020unified},  denoted by $\mathbf{W} \in \mathbb{R}^{M\times M}$, with ${w_{ij}}$ being the $(i,j)$-th element.
	To guarantee consensus, the matrix $\mathbf{W}$ is constrained to be a symmetric doubly stochastic matrix \cite{koloskova2019decentralized}. It is known that such a mixing matrix exists for every connected graph.

		\vspace{-1em}
	\subsection{MIMO IBFD Communication Channel}
	
	In each communication round, the learned model parameters of the devices are exchanged via wireless communication links as specified by $\mathcal{G}$.
	Each device is equipped with ${N_\text{T}}$ transmit antennas and ${N_\text{R}}$ receive antennas for full-duplex communication, yielding a  multiple-input multiple-output (MIMO) in-band full-duplex (IBFD) \emph{ad hoc} network with topology $\mathcal{G}$.\footnote{Here we consider IBFD communications where the exchange of model parameters between devices can be realized simultaneously.
		Our proposed scheme, as well as the subsequent analysis, can be readily extended to the half-duplex scenario by assuming that each device sequentially acts as a central server to perform over-the-air aggregation in a time-division fashion.
	}
	We further assume that the transmit and receive antennas for each device are well isolated, where the residual self-interference can be efficiently suppressed by using
	the self-interference cancellation (SIC) technique\cite{sabharwal2014band}.

	In each communication round $t$, each device broadcasts its local model parameter via multicast beamforming, and simultaneously receives the learned models from the neighbor devices. We assume a block-fading channel, i.e., the channel coefficients keep invariant within each communication round.  The received signal of each device at the $l$-th channel use, denoted by $\mathbf{y}_{i}^{(t)}[l] \in\C^{{N_\text{R}}} $, is given by
				\vspace{-1em}
	\begin{align}
		\mathbf{y}_{i}^{(t)}[l]=\sum_{j\in\cM_i}^{}\mathbf{H}_{\langle i,j\rangle}^{(t)}\mathbf{s}_{j}^{(t)}[l]+\mathbf{n}_{i}^{(t)}[l],~ \forall i \in [M],\label{channel_oneuse}
	\end{align}
	where $\cM_i$ denotes the neighbor set of device $i$, $\mathbf{H}_{\langle i,j\rangle}^{(t)}\in \C^{{N_\text{R}}\times {N_\text{T}}}$ denotes the channel matrix between the $i$-th device and the $j$-th device,  $\mathbf{s}_{j}^{(t)}[l] \in\C^{{N_\text{T}}}$ denotes the transmit signal of user $j$ in the $l$-th channel use, and $\mathbf{n}_{i}^{(t)}[l] \in\C^{{N_\text{R}}}$ is an additive white Gaussian noise (AWGN) vector with each element following the distribution $\mathcal{CN}(0,\sigma_n^2)$.
	Let $L$ be the number of channels used in each communication round. Then, the received signal matrix of each device  can be expressed as
				\vspace{-1em}
	\begin{align}
		\mathbf{Y}_{i}^{(t)}=\sum_{j\in\cM_i}^{}\mathbf{H}_{\langle i,j\rangle}^{(t)}\mathbf{S}_{j}^{(t)}+\mathbf{N}_{i}^{(t)},~ \forall i \in [M],\label{channel}
	\end{align}
	where $\mathbf{Y}_{i}^{(t)}\triangleq\left[\mathbf{y}_{i}^{(t)}[1],\cdots,\mathbf{y}_{i}^{(t)}[L]\right]\in \C^{{N_\text{R}}\times L}$,  $\mathbf{S}_{j}^{(t)}\triangleq\left[\mathbf{s}_{j}^{(t)}[1],\cdots,\mathbf{s}_{j}^{(t)}[L]\right] \in \C^{{N_\text{T}}\times L} $ 
	and  $\mathbf{N}_{i}^{(t)}\triangleq\left[\right.\nn_{i}^{(t)}[1]\cdots,\nn_{i}^{(t)}[L]\left.\right]\in \C^{{N_\text{R}}\times L}$.
	We assume that the global channel state information (CSI) is available. In practice, CSI can be obtained by using conventional channel estimation techniques and exploiting channel reciprocity and/or effective  feedback\cite{zhu2018mimo}, \cite{wen2014channel}.
	\section{Proposed MIMO OA-DFL Framework}
	In this section, we illustrate the proposed MIMO OA-DFL framework. Specifically, in each training round, each device computes the local gradient and then performs gossip model aggregation over the channel given in \eqref{channel} based on over-the-air computation. After that, each device updates its local model according to \eqref{update}. In the following, we focus on the over-the-air aggregation process.

	To begin with, in over-the-air aggregation,  each device needs to simultaneously broadcast its local model parameter using the same frequency resource via multicast beamforming and analog domain modulation. By cooperatively controlling the multicast transmit and receive beamformers, the expected aggregation signal can be coherently recovered at each device\footnote{In a decentralized (\emph{ad hoc}) system, 
		to guarantee the synchronization of arriving signal,  all the devices need to be synchronized by a unified clock\cite{romer2001time}. As an example, the cyclic prefix (CP) technique, originally used in orthogonal frequency-division multiplexing (OFDM) systems, can be exploited for signal synchronization\cite{sandell1995timing}. }. To be specific,
	at an arbitrary communication round, the following procedure is concurrently executed on every device.   
	We first normalize the model parameter $\xx_i^{(t)}$ as
	\begin{align}\label{nomalize}
		\tilde{\xx}_{i}^{(t)}={\left(\xx_i^{(t)}-\bar{x}_{i}^{(t)}\1_D\right)}/\sqrt{{v_{i}^{(t)}}},~\forall i \in [M]
	\end{align}
	where $ \bar{x}_{i}^{(t)}=\frac{1}{D}\sum_{d=1}^{D}{x}_{i}^{(t)}[d]$ and ${v_{i}^{(t)}}=\frac{1}{D}\sum_{d=1}^{D}\left({x}_{i}^{(t)}[d]-\bar{x}_{i}^{(t)}\right)^2$ are the mean and variance of ${\xx}_{i}^{(t)}$, respectively. By following the common practice, e.g., in \cite{liu2021reconfigurable} and \cite{lin2021deploying}, the mean and variance are exchanged between the neighbors via error-free links. In this normalization process, the model parameter $\xx_i^{(t)}$ is transformed into a zero-mean and unit-variance signal $\tilde{\xx}_{i}^{(t)}$. 
	Then, we  convert the normalized model vector
	$\tilde{\xx}_{i}^{(t)} \in \R^D$  to a complex version $\rr_{i}^{(t)} \in \C^L$ 
	\begin{align}\label{complex1}
		\rr_{i}^{(t)}= \tilde{\xx}_{i}^{(t)}\left(1:\frac{D}{2}\right)+\text{j}\tilde{\xx}_{i}^{(t)}\left(\frac{D+2}{2}:D\right),~\forall i \in [M],
	\end{align}
	where we choose the block length $L=D/2$ for simplicity. Let $\uu_{i}^{(t)} \in \C^{{N_\text{T}}}$ be the multicast beamforming vector.  The transmit signal of the $i$-th device, denoted by $\mathbf{S}_{i}^{(t)}$, can be expressed as
				\vspace{-1em}
	\begin{align}
		\mathbf{S}_{i}^{(t)}\triangleq\uu_{i}^{(t)}({\rr_{i}^{(t)}})^{\mathrm{T}}\in \C^{{N_\text{T}}\times L},\label{trans_signal}
	\end{align}
	and the corresponding the power constraint is $~\E\norm{\mathbf{S}_{i}^{(t)}[l]}^2=2\norm{\uu_{i}^{(t)}}^2\leq P_0, \forall i \in [M]$, where $P_0$ denotes the maximum transmit power for each device and $\mathbf{S}_{i}^{(t)}[l]\triangleq r_{i}^{(t)}[l]\uu^{(t)}_j \in \C^{N_\text{T}}$  is the transmit signal of device $i$ at the $l$-th channel use. Then, each device broadcasts the signal $\mathbf{S}_{i}^{(t)}$ through the channel given in \eqref{channel} to its neighbors. The received signal of each device can be expressed as
	\begin{align}\label{wholechannel}
		\hat{\rr}_{i}^{(t)}=\Big(({\ff_{i}^{(t)}})^{\mathrm{H}}\mathbf{Y}_{i}^{(t)}\Big)^{\mathrm{T}}=\Big(\sum_{j\in\cM_i}\rr_{j}^{(t)}(\mathbf{H}_{\langle i,j\rangle}^{(t)}\uu_{j}^{(t)})^{\mathrm{T}}+\mathbf{N}_{k,i}^{\mathrm{T}}\Big)({\ff_{i}^{(t)}})^{\ast},~\forall i \in [M]
	\end{align}
	where $\ff_{i}^{(t)}\in \C^{{N_\text{R}}}$ represents the receive beamforming (combining) vector used to retrieve the desired signal. Then, each device computes the estimate of $\xx_{i}^{(t + \frac{1}{2})}$  from $\hat{\rr}_{i}^{(t)}$ by  
	\begin{align}\label{rece1}
		\hat{\xx}_{i}^{(t+\frac{1}{2})}=\left[\mathrm{Re}\{\hat{\rr}_{i}^{(t)}\}^{\mathrm{T}},~\mathrm{Im}\{\hat{\rr}_{i}^{(t)}\}^{\mathrm{T}}\right]^{\mathrm{T}}+\tilde{x}_{i}^{(t)}\1_{D}+w_{i,i}\xx_i^{(t)},~\forall i \in [M]
	\end{align}
	where $\tilde{x}_{i}^{(t)}\triangleq\sum_{j\in \cM_i}{w_{ij}}\bar{x}_{j}^{(t)}$. Note that  the  term $\tilde{x}_{i}^{(t)}\1_{D}$ is added back to compensate the mean of $\xx_i^{(t)}$  subtracted in the normalization step \eqref{nomalize},
	and the term  $w_{i,i}\xx_i^{(t)}$ represents the contribution of local model $\xx_i^{(t)}$ to the model aggregation.
	
	With the collected received signal $\hat{\xx}_{i}^{(t+\frac{1}{2})}$, each device updates the local model based on \eqref{update}.
	We summarize the overall MIMO OA-DFL scheme in Algorithm \ref{alg:FL_framework}, where $\ff^{(t)} \triangleq \{\ff_{i}^{(t)}\}_{i=1}^M$ and $\uu^{(t)} \triangleq \{\uu_{i}^{(t)}\}_{i=1}^M$ are introduced for notational brevity.
	\begin{algorithm}[htb]
		\caption{MIMO OA-DFL scheme} 
		\label{alg:FL_framework} 
		\begin{algorithmic}[1] 
			\REQUIRE Training round $T$, data distribution $\{{\mathcal D}_i\}_{i=1}^{M}$. 
			\STATE {\textbf{Initialization:}} $t = 0$, the initial model $\{\mathbf{x}^{(0)}\}$ on the each device.
			\FOR{ $t \in [T]$ }
			\STATE Devices obtain the CSI and optimize $(\mathbf{W}, \ff^{(t)}, \uu^{(t)})$;
			\STATE Each device exchanges the mean  $\{\bar{x}^{(t)}\}^{M}_{i=1}$ and variance $\{v^{(t)}\}_{i=1}^{M}$ with their neighbors via error-free links;
			\FOR{$i \in [M]$ in parallel}
			\STATE  Device $i$ computes its local gradient  $ \nabla F(\xx_i^{(t)}, \xi_i^{(t)})$ by randomly sampling $\xi_i^{(t)}$ in local dataset;
			\STATE  Device $i$ broadcasts its local model $\{\xx_i^{(t)}\}$ to the neighbor devices via \eqref{nomalize}-\eqref{trans_signal};
			\STATE Device $i$ recovers the aggregated model  $\{\xx_{i}^{(t + \frac{1}{2})}\}$  based on \eqref{wholechannel} and \eqref{rece1};
			\STATE Device $i$  updates the local model $\{\xx_{i}^{(t+1)}\}$ based on \eqref{update};
			\ENDFOR
			\ENDFOR 
		\end{algorithmic}
	\end{algorithm}
	
	In the proposed MIMO OA-DFL scheme, model consensus is accomplished via D2D  communications. The existence of communication errors makes the learned model inaccurate and
	even compromises the consensus performance of MIMO OA-DFL. This poses a great challenge for the system design. In the next section, we analyze the convergence of MIMO OA-DFL and study
	the impact of the mixing matrix $\mathbf{W}$ and the beamformers $\ff^{(t)} $ and $\uu^{(t)}$ on the performance of MIMO OA-DFL.

	\section{Convergence Analysis}\label{sec_ca} \label{sec-performance-analysis}
	
	\subsection{Assumptions}
	To begin with, we make the following assumptions.
	\begin{assumption}\label{as1}
		(\rm{Gossip matrix}). 
		The mixing matrix $\mathbf{W}$ is a symmetric doubly stochastic matrix, i.e., $\mathbf{W}^{\mathrm{T}}=\mathbf{W}$, $ \mathbf{W}\1=\1$, $\1^{\mathrm{T}}\mathbf{W}=\1^{\mathrm{T}}$ and $\mathbf{W} \in [0,1]^{M\times M}$.
		We define	${\delta(\mathbf{W})} \triangleq (\max\{|\lambda_2(\mathbf{W})|, |\lambda_M(\mathbf{W})|\})^2$ and assume ${\delta(\mathbf{W})} < 1$.
	\end{assumption}
	
	\begin{assumption}\label{as2}
		($\omega$-\rm{smoothness}). The functions $f_1,\dots,f_M$ are all differentiable and the corresponding gradients $\nabla f_1(\cdot),\dots,f_M(\cdot)$ are Lipschitz continuous with parameter $\omega$, i.e.,
		\begin{align}
			\norm{\nabla f_i(\xx)-\nabla f_i(\yy)}\leq \omega \norm{\xx-\yy}, \forall \xx,\yy \in \R^{D},\forall i \in [M].
		\end{align}
	\end{assumption}
	
	\begin{assumption}\label{as3}
		(\rm{Bounded variance}). The variance of the stochastic gradient $\E{\norm{\nabla F(\xx, \xi_i) - \nabla f_i(\xx)}}^2$ and $	\E{\norm{\nabla f_i(\xx) - \nabla f(\xx)}}^2$ are bounded, i.e.,
		\begin{align}
			\E_{\xi_i\sim{\mathcal D}_i}{\norm{\nabla F(\xx, \xi_i) - \nabla f_i(\xx)}}^2 &\leq \alpha^2\,, \forall\xx \in \R^D\,,\forall i \in [M],\label{ass31}\\
			\E_{i\sim[M]}{\norm{\nabla f_i(\xx) - \nabla f(\xx)}}^2 &\leq  \beta^2\,,   \forall \xx \in \R^D.\label{ass32}
		\end{align}
		where $\alpha^2$  denotes the bound of the variance of stochastic gradients at each device, and $\beta^2$ denotes the bound of discrepancy of data
		distributions at different devices.
	\end{assumption}
	
	Assumptions 1-3 are commonly used in the literature on decentralized stochastic optimization and gossip algorithm; see, e.g., \cite{koloskova2019decentralized},\cite{wang2018cooperative},\cite{li2021decentralized}. Assumption 1 is related to the mixing matrix. Note that for a doubly stochastic matrix, we always have $\lambda_1(\mathbf{W})= 1$ and $|\lambda_i(\mathbf{W})|\leq 1,\forall i$.  Assumption 1 states that $\lambda_i(\mathbf{W})$ is strictly less than $1$ for $i\neq 1$. Later we see that ${\delta(\mathbf{W})}$  is related to the consensus performance in the decentralized network. 
	Assumption 2  is related to the Lipschitz continuity of the loss function.
	Assumption 3 ensures a bounded gap between the gradient of the local sample-dependent loss, i.e., $\nabla F(\xx, \xi_i)$, and that of the overall loss, i.e., $\nabla f(\xx)$. 
	\vspace{-1em}
	\subsection{Convergence Analysis of MIMO OA-DFL}
	To facilitate the analysis, we introduce the following lemma based on Assumption 1.
	\begin{lemma}\label{lemma_Wmatrix}
		For every $\mathbf{W}$ satisfying Assumption 1, we have
		\begin{align}
			\norm{\mathbf{W}^k -\frac{1}{M}\1\1^\mathrm{T}}_2^2 \leq {\delta(\mathbf{W})}^k, ~\forall k \in \R_+.
		\end{align}
	\end{lemma}
	\begin{proof}
		See \cite[Remark 15]{koloskova2019decentralized}.
	\end{proof} 
	{Lemma~\ref{lemma_Wmatrix}} states that  $\mathbf{W}^{k}$  converges to $\frac{1}{M}\1\1^\mathrm{T}$ in the sense of $\ell_2$ norm as $k$ goes to infinity.  Note that $\frac{1}{M}\1\1^\mathrm{T}$ is itself a  symmetric doubly stochastic matrix, representing a fully connected communication topology. 
	The global model average $\frac{\mX^{(t)}\1}{M}=\frac{1}{M}\sum_{i=1}^{M}{\xx}_{i}^{(t)}$ can be accessed by every device in this topology, which is similar to the centralized federated learning\cite{amiri2020federated}.
	\begin{proposition}\label{theorem1}
		Under Assumption 1-3, with $\lambda \leq {1}/{\omega}$, we have
		\begin{align}
			&\frac{1}{T}\sum_{t=0}^{T-1}\E\norm{\nabla f\left(\frac{\mX^{(t)}\1}{M}\right)}^2\leq \frac{1}{\left(\frac{1}{2}-{27M\lambda^2}G(\mathbf{W})\right)}\Bigg( \frac{f(\frac{\mX^{(0)}\1}{M})-{f^{\star}}}{\lambda T }+\frac{\alpha^2}{M}\notag\\&
			+(3M\alpha^2\lambda^2+27M\beta^2\lambda^2)G(\mathbf{W})+\frac{9G(\mathbf{W})}{T}\sum_{t=0}^{T-1}\E\norm{\mE^{(t)}}^2_F+\frac{1}{\lambda^2{M}^2T}\sum_{t=0}^{T-1}\E\norm{{\mE^{(t)}\1}}^2\Bigg)\label{conver1}
		\end{align}
		where the expectation on the left hand side of \eqref{conver1} is over the randomness of channel noise and stochastic data sampling, the expectation on the right hand side is over the randomness of channel noise, $\frac{1}{T}\sum_{t=0}^{T-1}\E\norm{\nabla f\left(\frac{\mX^{(t)}\1}{M}\right)}^2$  is the convergence metric\cite{lian2017can}, the right hand side of   \eqref{conver1} is the convergence bound, $G(\mathbf{W})\triangleq\frac{\omega^2}{(1-\sqrt{{\delta(\mathbf{W})}})^2-27M\lambda^2\omega^2}$, ${\mathbf{X}}^{(t)} \triangleq \left[{\xx}_{1}^{(t)},\dots, {\xx}_{M}^{(t)}~\right] $, $\hat{\mX}^{(t+\frac{1}{2})}\triangleq \left[{\hat{\xx}}_{1}^{(t+\frac{1}{2})},\dots, \hat{\xx}_{M}^{(t+\frac{1}{2})}~\right]$, $\mE^{(t)}\triangleq \mX^{(t)}\mW-\hat{\mX}^{(t+\frac{1}{2})} $ denotes the communication error matrix for all devices in round $t$, and ${f^{\star}}$ denotes the minimum value of the loss function.
	\end{proposition}	
	\begin{proof}
		Please refer to Appendix \ref{Appendix_Prf_covergence1}.
	\end{proof}
	Since $\frac{\mX^{(t)}\1}{M}=\frac{1}{M}\sum_{i=1}^{M}{\xx}_{i}^{(t)}$, the above proposition captures the convergence of the average of local model ${\xx}_{i}^{(t)}$, considering that there is no unified model among the decentralized devices\footnote{In Section \ref{sec-simluation}, we show that all the devices can reach consensus under our design.}. 
	
	To simplify our analysis, for each communication round $t$, we assume that the model parameters  $\{\tilde{\xx}^{(t)}_i| i \in [M]\}$ are independent and the model parameter elements $\{\tilde{x}^{(t)}_i[d]| d \in [D]\}, \forall i \in [M]$ are independent and identically distributed. Then, we have the following correlation matrices
	\begin{align}
		\E\left[ \tilde{\xx}^{(t)}_i (\tilde{\xx}^{(t)}_j)^{\text{T}} \right]=\mathbf{0}, \forall i \neq j \in [M],~\text{and} ~\E\left[ \tilde{\xx}^{(t)}_i (\tilde{\xx}^{(t)}_i)^{\text{T}} \right]=\mathbf{I},\forall i\in [M].\label{corr}
	\end{align}
	Based on the above assumption, we have the following proposition.
	\begin{proposition}
		\label{Corollary:MSE}
		Under the MIMO OA-DFL scheme, with the correlation assumption given in \eqref{corr}, the terms related to communication error matrix $\mE^{(t)}$ in \eqref{conver1} are given by
		\begin{align}
			\E\norm{\mE^{(t)}}^2_F
			={C}&\sum_{p=1}^{M}\bigg(\sum_{i\in M_p}2\left(w_{ip}{v_{p}^{(t)}}\right)^2-4\sum_{i\in M_p}w_{ip}\operatorname{Re}\left\{{v_{p}^{(t)}}({\ff_{i}^{(t)}})^H{\uu_{p}^{(t)}}{\mathbf{H}_{\langle i,p\rangle}^{(t)}}\right\}\notag\\&+2\sum_{i\in M_p}\left({\ff_{i}^{(t)}})^{\mathrm{H}}{\mathbf{H}_{\langle i,p\rangle}^{(t)}}{\uu_{p}^{(t)}}\right)\left(({\ff_{i}^{(t)}})^{\mathrm{H}}{\mathbf{H}_{\langle i,p\rangle}^{(t)}}{\uu_{p}^{(t)}}\right)^H+\sigma_n^2\norm{{\ff_{i}^{(t)}}}^2\bigg)\label{E_t}\\
			\E\norm{{\mE^{(t)}\1}}^2
			=&\frac{C{M}^2}{n^2}\bigg(\sum_{p=1}^{M}\sum_{i,j\in \cM_p}^{}2(w_{ip}w_{jp}({v_{p}^{(t)}})^2)-4\sum_{p=1}^{M}\sum_{i,j\in \cM_p}^{}w_{ip}\operatorname{Re}\left\{{v_{p}^{(t)}}({\ff_{j}^{(t)}})^{\mathrm{H}}\mathbf{H}_{\langle j,p\rangle}^{(t)}{\uu_{p}^{(t)}}\right\}\notag\\
			&+\!\!2\sum_{p=1}^{M}\sum_{i,j\in \cM_p}^{}\left(({\ff_{j}^{(t)}})^{\mathrm{H}}\mathbf{H}_{\langle j,p\rangle}^{(t)}{\uu_{p}^{(t)}}\right)\left(({\ff_{i}^{(t)}})^H{\mathbf{H}_{\langle i,p\rangle}^{(t)}}{\uu_{p}^{(t)}}\right)
			\!+\!\sum_{i=1}^{M}\left(\sigma_n^2\norm{{\ff_{i}^{(t)}}}^2\right)\!\!\bigg)\label{E_t_1}
		\end{align}
	\end{proposition}
	\begin{proof}
		Please refer to Appendix \ref{app_c}.
	\end{proof}
	
	\begin{proposition}
		\label{monotonic}
		The right hand side (RHS) of \eqref{conver1}  monotonically increases with respect to  ${\delta(\mathbf{W})}$.
	\end{proposition}
	\begin{proof}
		The RHS of \eqref{conver1} can be abbreviated as  $f(G(\mathbf{W}))=\frac{A+G(\mathbf{W})C}{1/2-G(\mathbf{W})D}$, where  $A,B,C,D\geq 0$. Note that $f'(G(\mathbf{W}))=\frac{1/2C+AD}{(1/2-GD)^2}\geq 0$, implying that $f(G(\mathbf{W}))$ is monotonically increasing with respect to $G(\mathbf{W})$.
		Furthermore, since $0\leq\delta(\mathbf{W})<1$ (by Assumption 1), $G(\mathbf{W})$ is also monotonically increasing with respect to $\delta(\mathbf{W})$. Therefore, we conclude that the RHS of \eqref{conver1} monotonically increases with respect to $\delta(\mathbf{W})$. 
	\end{proof}

	\begin{remark}
		{Proposition \ref{theorem1}} provides some insights on the convergence of MIMO OA-DFL. 
		From the communication perspective, it can be observed that the existence of communication error  $\mE^{(t)},\forall t \in [T]$  reduces the convergence rate, where both error terms $\E\norm{\mE^{(t)}}^2_F$ and $\E\norm{{\mE^{(t)}\1}}^2$ accumulate over the training rounds and enlarge the convergence bound.
		From the learning perspective, as shown in {Proposition \ref{monotonic}}, the value of the second-largest squared eigenvalue ${\delta(\mathbf{W})}$   plays a critical role on the learning accuracy. This indicates that the mixing matrix $\mW$ needs to be designed to achieve smaller ${\delta(\mathbf{W})}$ for fast convergence\footnote{We emphasize the determination of aggregation weight in MIMO OA-DFL is different from conventional FL. In FL, the aggregation weights are usually chosen according to the size of the local data set\cite{mcmahan2017communication}. But for the MIMO OA-DFL system, the mixing matrix must satisfy the symmetric doubly stochastic constraint to guarantee consensus,  and need to be carefully designed to improve convergence performance.}. 
	\end{remark}
	From {Propositions \ref{theorem1}} and {\ref{Corollary:MSE}}, we see that the mixing matrix $\mW$ and the beamformers  \{$\uu^{(t)},\ff^{(t)}$\} jointly have impact on the learning performance. In the following, we propose a systematic	communication (i.e.,  beamformers) and learning (i.e.,  mixing matrix)  co-design algorithm to improve the performance of the MIMO OA-DFL system.

	\section{System Optimization}\label{sys_opt}
	To achieve a better learning performance in MIMO OA-DFL, we propose to minimize the RHS of \eqref{conver1}  over   $\mW$,  $\uu^{(t)}$ and  $\ff^{(t)}$. The details are provided below.
		\vspace{-1em}
	\subsection{Problem Formulation}
	We design the MIMO OA-DFL system to minimize the convergence bound \eqref{conver1}. 
	We conduct the system optimization in a round-by-round fashion. For a given decentralized topology,  we minimize the round-based convergence bound by jointly optimizing the mixing matrix $\mW^{(t)}$, the multicast beamformers $\uu^{(t)}$ and the receive beamformers  $\ff^{(t)}$. We omit the superscript $t$ in the sequel for brevity. The optimization problem is then cast as
	\begin{subequations}\label{eq28}
		\begin{align}
			{\text{(P1)}:}~\min_{\mW,\ff,\uu}& \quad \Psi(\mW,\ff,\uu)\triangleq\frac{\left( Q
				+RG(\mathbf{W})+9G(\mathbf{W})\E\norm{\mE}^2_F+\frac{1}{\lambda^2{M}^2}\E\norm{{\mE\1}}^2\right)}{\left(\frac{1}{2}-{27M\lambda^2}G(\mathbf{W})\right)}\label{constop1}\\
			\text{s.t.}&\quad w_{ij}=0,\forall \{ij\}\not\in \mathcal{E},\mathbf{W}^{\mathrm{T}}=\mathbf{W},  \mathbf{W}\1=\1, \mathbf{W} \in [0,1]^{M\times M},\\
			&\quad\norm{\uu_{i}}^2\leq P_0/2,\forall i \in [M],\label{ori_c}
		\end{align}
	\end{subequations}
	where $Q=\frac{f(\frac{\mX^{(0)}\1}{M})-{f^{\star}}}{\lambda T }+\frac{\alpha^2}{M}$, $R=3M\alpha^2\lambda^2+27M\beta^2\lambda^2$, and $G(\mathbf{W})=\frac{\omega^2}{(1-\sqrt{{\delta(\mathbf{W})}})^2-27M\lambda^2\omega^2}$.

	P1 is a non-convex problem. Different from the existing solutions \cite{shi2021over},\cite{lin2022distributed},\cite{yang2020federated} that the transceiver beamforming vectors can be optimized alternately, the new challenge is that even with given beamformers $\ff$ and $\uu$, problem P1 is still non-convex due to the coupling of $\mW$ and its the second-largest squared eigenvalue ${\delta(\mathbf{W})}$. However, by exploiting the monotonicity of  ${\delta(\mathbf{W})}$ and the structural information of matrix $\mW$, this problem can be efficiently solved in an AO manner, as detailed in what follows.
		\vspace{-1em}
	\subsection{Optimizing Beamformers for Given Mixing Matrix}
	We first optimize the beamforming vectors $\uu$ and $\ff$ for given  $\mW$. Dropping the irrelevant terms, we have the following problem
	\begin{align}\label{problem2}
		{\text{(P2)}:}~\min_{\ff,\uu}& \quad d(\mW,\ff,\uu)\triangleq 9G(\mathbf{W})\E\norm{\mE}^2_F+\frac{1}{\lambda^2{M}^2}\E\norm{{\mE\1}}^2,~
		\text{s.t.}~~\eqref{ori_c},
	\end{align}
	where $\E\norm{\mE}^2_F$ and $\E\norm{{\mE\1}}^2$ are given by \eqref{E_t} and \eqref{E_t_1}, respectively. We optimize  $\ff$ and $\uu$ in an alternating fashion, as detailed below.
	\subsubsection{Optimizing $\uu$ for fixed $\ff$} For a fixed $\ff$, the multicast beamforming vectors in $\uu$  can be determined by solving the following problem:
	\begin{align}
		\vspace{-1em}
		{\text{(P3)}:}~\underset{\uu}{\min}~~
		\sum_{p=1}^{M}{\uu_{p}^{\mathrm{H}}}\mathbf{M}_{p}{\uu_{p}}-2\operatorname{Re}\left\{\sum_{p=1}^{M}\mathbf{n}_{p}^{\mathrm{H}}{\uu_{p}}\right\}
		~~\text{s.t.~\eqref{ori_c}}.
	\end{align}
\vspace{-1em}
	where
	\begin{subequations}
			\vspace{-1em}
		\begin{align}
			&\mathbf{M}_{p}={9G(\mathbf{W})}\sum_{i\in \cM_p}{\mathbf{H}_{\langle i,p\rangle}^{\mathrm{H}}}{\ff_{i}}{\ff_{i}}^\mathrm{H}{\mathbf{H}_{\langle i,p\rangle}}+\frac{1}{{\lambda}^2}\frac{1}{M^2}\sum_{i,j\in \cM_p}^{}{\mathbf{H}_{\langle i,p\rangle}^{\mathrm{H}}}{\ff_{i}}{\ff_{j}}^\mathrm{H}{\mathbf{H}_{\langle j,p\rangle}},\label{M_p}\\
			&\mathbf{n}_{p}={9G(\mathbf{W})}\sum_{i\in \cM_p}w_{ip}{v_{p}}({\ff_{i}}^\mathrm{H}{\mathbf{H}_{\langle i,p\rangle}})^{\mathrm{H}}+\frac{1}{{\lambda}^2}\frac{1}{M^2}\sum_{i,j\in \cM_p}^{}{v_{p}}w_{ip}({\ff_{j}}^\mathrm{H}{\mathbf{H}_{\langle j,p\rangle}})^{\mathrm{H}}.\label{n_p}
		\end{align}
	\end{subequations}
	For the  term $\sum_{i,j\in \cM_p}^{}{\mathbf{H}_{\langle i,p\rangle}^{\mathrm{H}}}{\ff_{i}}{\ff_{j}}^\mathrm{H}{\mathbf{H}_{\langle j,p\rangle}}$ in $\mathbf{M}_{p}$, we have
	\begin{align}
		\xx^{\mathrm{H}} \bigg(\sum_{i,j\in \cM_p}^{}{\mathbf{H}_{\langle i,p\rangle}^{\mathrm{H}}}{\ff_{i}}{\ff_{j}}^\mathrm{H}{\mathbf{H}_{\langle j,p\rangle}}\bigg)\xx
		\!\!=\!\!\bigg(\sum_{i\in\cM_p}{}\xx^{\mathrm{H}}{\mathbf{H}_{\langle i,p\rangle}^{\mathrm{H}}}{\ff_{i}}\bigg)\bigg(\sum_{i\in\cM_p}{}\xx^{\mathrm{H}}{\mathbf{H}_{\langle i,p\rangle}^{\mathrm{H}}}{\ff_{i}}\bigg)^{\mathrm{H}}\geq 0, \forall \xx \in \mathbb{C}^{{N_\text{R}}}.
	\end{align}
	Hence, $\mathbf{M}_{p},\forall p\in[M]$ is a positive semidefinite matrix and therefore P3 is a convex QCQP problem. This problem can be solved efficiently by considering its dual:
	\begin{align}
		{\text{(P4)}:}~\underset{{\lambda_p}}{\min}~~
		-{\mathbf{n}}_{p}^{\mathrm{H}}(\mathbf{M}_{p}+\lambda_p\mathbf{I})^{\dagger}\mathbf{n}_{p}-{P_0\lambda_p}/{2}
		~~\text{s.t.~} \lambda_p \geq 0, \forall p \in [M].
	\end{align}
	Then the optimal beamformer is given by $\uu_p^{\star}=(\mathbf{M}_{p}+\lambda_p^{\star}\mathbf{I})^{\dagger}\mathbf{n}_{p},\forall p \in [M]$, where $\lambda_p^{\star}$ is the solution to P4.
	\subsubsection{Optimizing $\ff_p$  for fixed $\uu$ and $\{{\ff_{i}}\}_{i\neq p}$} 
	We optimize each $\ff_p$ alternatingly.
	With fixed $\uu$ and $\{{\ff_{i}}\}_{i\neq p}$, problem P2 reduces to
	\begin{align}
		{\text{(P5)}:}~\underset{\ff_{p}}{\min}~~
		\ff_{p}^{\mathrm{H}}\mathbf{A}_{p}\ff_{p}-4\operatorname{Re}\{\mathbf{b}_{p}^{\mathrm{H}}\ff_{p}\}
	\end{align}
	\setlength{\belowdisplayskip}{0pt}
	where
	\begin{subequations}
			\vspace{-1em}
		\begin{align}
			&\mathbf{A}_{p}=({18G(\mathbf{W})}+2\frac{1}{{\lambda}^2}\frac{1}{M^2})\sum_{j\in \cM_p}\mathbf{H}_{\langle p,j\rangle}{\uu_{j}}{\uu_{j}^{\mathrm{H}}}\mathbf{H}_{\langle p,j\rangle}^{\mathrm{H}}+(\frac{1}{{\lambda}^2}\frac{1}{M^2}+{9G(\mathbf{W})})\sigma_n^2\mathbf{I}_{{N_\text{R}}},\label{A_p}\\
				\vspace{-1em}
			&\mathbf{b}_{p}={9G(\mathbf{W})}\!\!\sum_{j\in \cM_p}\!\!w_{pj}{v_{j}}({\uu_{j}^{\mathrm{H}}}\mathbf{H}_{\langle p,j\rangle}^{\mathrm{H}})^{\mathrm{H}}+\frac{1}{{\lambda}^2}\frac{1}{M^2}\sum_{i=1}^{M}\!\sum_{j\in \cM_i,\cM_p}^{}\!\!\!\!w_{ij}{v_{j}}\left({\uu_{j}^{\mathrm{H}}}\mathbf{H}_{\langle p,j\rangle}^{\mathrm{H}}\right)^{\mathrm{H}}\!\!\notag\\
				\vspace{-1em}
			&-\!\frac{1}{{\lambda}^2}\frac{1}{M^2}\sum_{i=1,i\neq p}^{n}\sum_{j\in \cM_p,\cM_i}^{}\!\!({\ff_{i}}^\mathrm{H}{\mathbf{H}_{\langle i,j\rangle}}{\uu_{j}}{\uu_{j}^{\mathrm{H}}}\mathbf{H}_{\langle p,j\rangle}^{\mathrm{H}})^{\mathrm{H}}.\label{b_p}
		\end{align}
	\end{subequations}
	This is an unconstrained convex problem, and the optimal solution is 
	$\ff_{p}^\star=2\mathbf{A}_{p}^{-1}\mathbf{b}_{p},\forall p \in [M]$.
	
		\vspace{-1em}
	\subsection{Optimizing Mixing Matrix for Given Beamformers}
	What remains is to optimize the mixing matrix. For given $\uu$ and $\ff$, the problem P1 can be expressed as
	\begin{subequations}\label{sub21}
		\begin{align}
			{\text{(P5)}:}~\min_{\mW}&\quad \frac{ Q
				+RG(\mathbf{W})+9G(\mathbf{W})\E\norm{\mE}^2_F+\frac{1}{\lambda^2{M}^2}\E\norm{{\mE\1}}^2}{\frac{1}{2}-{27M\lambda^2}G(\mathbf{W})}\label{28aaa}\\
			\text{s.t.}&\quad w_{ij}=0,\forall \{ij\}\not\in \mathcal{E},\mathbf{W}^{\mathrm{T}}=\mathbf{W},  \mathbf{W}\1=\1, \mathbf{W} \in [0,1]^{M\times M},\label{sub21_cons}
		\end{align}
	\end{subequations}
	where  $G(\mathbf{W})=\frac{\omega^2}{(1-\sqrt{{\delta(\mathbf{W})}})^2-27M\lambda^2\omega^2}$. 
	We introduce slack variable $\hat{\delta}$ and reformulate problem P5 as
	\begin{subequations}\label{sub22}
		\begin{align}
			{\text{(P6)}:}~\min_{\mW,\hat{\delta}}&\quad \frac{Q
				+R{G}(\hat{\delta})+9{G}(\hat{\delta})\E\norm{\mE}^2_F+\frac{1}{\lambda^2{M}^2}\E\norm{{\mE\1}}^2}{\frac{1}{2}-{27M\lambda^2}{G}(\hat{\delta})}\label{sub22_ojb1}\\
			\text{s.t.}&\quad {\delta(\mathbf{W})}\leq \hat{\delta},~ \eqref{sub21_cons}.
		\end{align}
	\end{subequations}
	where ${G}(\hat{\delta})=\frac{\omega^2}{(1-\sqrt{\hat{\delta}})^2-27M\lambda^2\omega^2}$.
	\begin{proposition}\label{equivalent}
		Problem P6 is equivalent to P5.
	\end{proposition}
	\begin{proof}
		From {Proposition~\ref{monotonic}},  the objective function \eqref{sub22_ojb1} monotonically increases  with respect to $\hat{\delta}$. So $\hat{\delta}$ can always be decreased to reduce the objective value, and consequently the constraint ${\delta(\mathbf{W})}\leq \hat{\delta}$ must hold with equality at the optimal point of P6. Therefore, problem P6 is equivalent to P5 without loss of optimality.
	\end{proof}
	We now optimize $\mW$ and $\hat{\delta}$ in an alternating manner.
	\subsubsection{Optimizing $\mW$ for fixed $\hat{\delta}$} For a fixed $\hat{\delta}$, the problem P6 reduces to
	\begin{align}
		{\text{(P7)}:}~\min_{\mW}&\quad \hat{d}(\mW,\ff,\uu)\label{sub22_ojb} ~~\text{s.t.~}\quad {\delta(\mathbf{W})}\leq \hat{\delta},~ \eqref{sub21_cons}.
	\end{align}
	where $\hat{d}(\mW,\ff,\uu)=9{G}(\hat{\delta})\E\norm{\mE}^2_F+\frac{1}{\lambda^2{M}^2}\E\norm{{\mE\1}}^2$.
	\begin{proposition}\label{sub_problem21}
		Problem P7 is a convex problem, which can be efficiently solved by e.g., interior-point method.
	\end{proposition}
	\begin{proof}
		Please refer to Appendix \ref{app_sub_problem21}.
	\end{proof}
	\subsubsection{Optimizing $\hat{\delta}$ with fixed $\mW$} With  fixed $\mW$,
	due to the monotonicity of $\hat{\delta}$ in the objective function P6, $\hat{\delta}$ can be directly updated by $\hat{\delta}={\delta(\mathbf{W})}$ in each iteration.
		\vspace{-1em}
	\subsection{Overall Algorithm for Optimizing  $\{\mW,\ff,\uu\}$}
	We summarize the proposed algorithm for optimizing  $\{\mW,\ff,\uu\}$ as Algorithm 2.

	\begin{algorithm}[h]
		\caption{ AO Algorithm for Optimizing $\{\mW,\ff,\uu\}$} 
		\label{alg_beamforming} 
		\begin{algorithmic}[1] 
			\REQUIRE  \{$\mathcal{M}_i,i \in [M] $\},$\{ \mathbf{H}_{\langle i, j \rangle},  |  i \in [M], j \in [M]\}$, $J_{\max}$, $I_{1{\max}}$ and $I_{2{\max}}$. 
			\STATE {\textbf{Initialization:}}  $\mathbf{f}$, $\mathbf{u}$ and $\mW$.
			\FOR{ $j \in [J_{\max}]$ }
			\FOR{$i_1 \in [I_{1{\max}}]$}
			\STATE Compute $\mathbf{M}=\{\mM_p\}_{p=1}^M$ and $\mathbf{n}=\{\nn_p\}_{p=1}^M$ based on \eqref{M_p} and \eqref{n_p}
			\STATE Optimize $\mathbf{u}=\{\uu_p\}_{p=1}^M$,  by solving (P4);
			\FOR{$p \in [M]$ }
			\STATE Compute $\mathbf{A}_p$ and $\mathbf{b}_{p}$ based on \eqref{A_p} and \eqref{b_p} ; 
			\STATE Update $\ff_{p}$ by the closed-form solution $\ff_{p}^\star=2\mathbf{A}_{p}^{-1}\mathbf{b}_{p}$
			\ENDFOR
			\ENDFOR
			\FOR{$i_2 \in [I_{2{\max}}]$}
			\STATE Optimize $\mW$ by solving (P7)
			\STATE Updates $\hat{\delta}$ based on $\hat{\delta}={\delta(\mathbf{W})}$;
			\ENDFOR
			\ENDFOR
			\ENSURE $\{\mW,\mathbf{f}, \mathbf{u}\}$.
		\end{algorithmic}
	\end{algorithm}
		
	Note that when executing Algorithm 2, we do not need to estimate parameters $Q$ and $R$  as defined in problem (P1), which simplifies our algorithm.
	Furthermore, the weights of error terms $\E\norm{\mE}^2_F$ and $E\norm{{\mE\1}}^2$ in problem (P1) are based on hypothetical parameters. It may be challenging to estimate the appropriate parameters for each specific MIMO OA-DFL scenario. To enhance the robustness of the algorithm, we can use $(\frac{1}{{\lambda}^2M}+{9G(\mathbf{W})})\E\norm{\mE}^2_F$ to substitute $9G(\mathbf{W})\E\norm{\mE}^2_F+\frac{1}{\lambda^2{M}^2}\E\norm{{\mE\1}}^2$ in (P1) by noting $\E\norm{{\mE\1}}^2\leq {M}\E\norm{\mE}^2_F$.
	
	We now provide a concise discussion of the computational complexity associated with Algorithm 2. In this algorithm, both problem (P4) and problem (P7) are convex problems, making them amenable to solution using existing optimization solvers based on interior-point methods. Consequently, the worst-case complexity of Algorithm 2 can be expressed as
	$\mathcal{O}(J_{max}(I_{1max}M{N}^{3.5}+I_{2max}M^7))$, where  
	$N =N_\text{T}$ denotes the number of transmit antennas of each device,
	$J_{\max}$ denotes the maximum iteration times for  Algorithm 2, $I_{1{\max}}$ represents the maximum iteration times for solving the beamformers optimization subproblem (as described in Section \ref{sys_opt}-B), and $I_{2{\max}}$ signifies the maximum iteration times for solving the mixing matrix optimization subproblem (as described in Section \ref{sys_opt}-C).

	\section{Simulation Results}\label{sec-simluation}
	\subsection{Simulation Under Error Free Case}
	To start with, we conduct experiments to verify the convergence result in {Proposition 1}.  To analyze the impact of the second-largest squared eigenvalue of mixing matrix on the system performance, we consider an error-free case and perform DFL training with different mixing matrices. The training process  is  illustrated in Section \ref{section-system-model}-A where we have $\hat{\xx}_{i}^{(t+\frac{1}{2})} = \xx_{i}^{(t + \frac{1}{2})},\forall i \in [M]$ in \eqref{update}.
	
	We perform the learning task of image classification on the MNIST dataset \cite{deng2012mnist}.  We use 20k samples to train the model and 10k samples for validation from the original data set. The heterogeneous data splitting scheme in \cite{mcmahan2017communication} is implemented. 
	To be specific, there are 10 classes in the MNIST dataset so we divide the devices into 10 equally sized groups, with each group of devices evenly assigned disjoint data samples from a specific class.
	For the network configuration, we train a convolutional neural network (CNN) with two $5\times5$ convolution layers (separately with 10 and 20 channels and each followed by $2\times2$ max pooling), a subsequent batch normalization layer, a fully connected layer containing 50 units with ReLu activation and a final softmax output layer. The network has 21880 parameters in total. The cross-entropy loss is used as the loss function.
	
	In Fig.~\ref{fig_errorfree}, we plot the minimum test accuracy (among all devices) and the average test accuracy (of the global model average) with different choices of the mixing matrix over 150 communication rounds. We randomly generate the different mixing matrices satisfying Assumption 1 by using the convex optimization tool CVXPY\cite{diamond2016cvxpy}. The mixing matrix with ${\delta(\mathbf{W})} =0$ corresponds to the fully connected structure where the value of each element is $1/M$.  We set the number of devices $M=30$, learning rate $=0.02$,  momentum $=0.9$ and the results are averaged over 30 Monte Carlo trials.
	\begin{figure}[htbp]
		\vspace{-1em}
		\centering
		\includegraphics[width=1\linewidth]{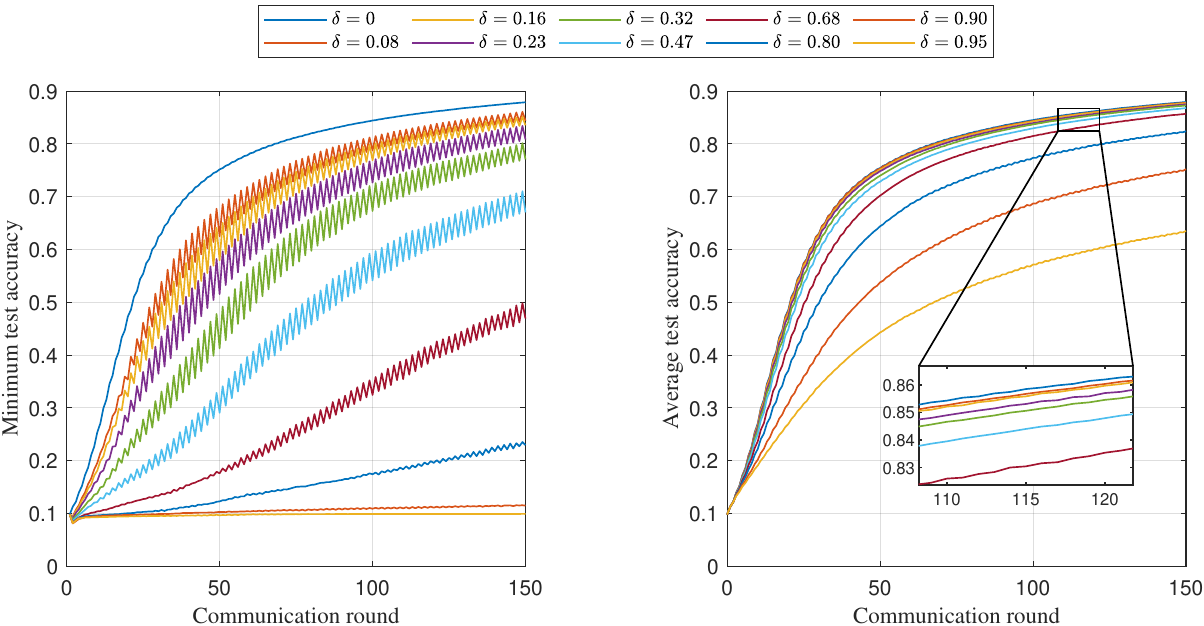}
			\vspace{-1em}
		\caption{Minimum and average test accuracy versus communication round for different choices of the mixing matrix.}
		\label{fig_errorfree}
		\vspace{-1em}
	\end{figure}
	
	As illustrated in Fig.~\ref{fig_errorfree}, we observe that the test accuracy gradually deteriorates as the increase of ${\delta(\mathbf{W})}$ in both subgraphs, which matches our analysis in {Proposition~\ref{theorem1}} well.  For the test accuracy of the global model average (right subgraph), the accuracy gaps between adjacent ${\delta(\mathbf{W})}$ are relatively narrow, especially when the value of ${\delta(\mathbf{W})}$ is small (less than 0.32). However, these gaps become larger in terms of the minimum test accuracy (left subgraph). In the left subgraph, we see that only the minimum accuracy of ${\delta(\mathbf{W})}=0$ (fully connected) can keep close to the accuracy curve of the global model average. For  ${\delta(\mathbf{W})}$ more than $0.8$, the worst-case learning performance in the left subgraph is prominently poor (less than 0.3). 
	This is because for the system with high ${\delta(\mathbf{W})}$, there are significant discrepancies among the local models, resulting in extremely poor performance for some devices.
	Therefore, the second-largest squared eigenvalue ${\delta(\mathbf{W})}$ has a significant impact on the consensus performance.
	
	\vspace{-1em}
	\subsection{Performance of Proposed Algorithm Under Various Settings}
	In this subsection, we study the performance of the proposed algorithm in different network topologies and communication configurations. 
	We utilize the sparsity level of the mixing matrix as a characterization metric for different network topologies. The sparsity level is determined by the proportion of absent communication links, expressed as the ratio of the number of zero elements to the total number of elements in the mixing matrix, i.e., $\frac{\text{number of 0 elements}}{M^2}$. To create different network topologies, we randomly generate the corresponding number of zero elements in the mixing matrix.
	By employing this approach, we obtain network topologies with different sparsity levels and compare the performance of the proposed algorithm under four specific sparsity levels: 0\%, 30\%, 60\%, and 90\%. A sparsity level of 0\% corresponds to a fully connected topology, where all communication links are present. A sparsity level of 30\% encompasses topologies with relatively dense communication links. A sparsity level of 60\% covers relatively sparse topologies, and a sparsity level of 90\% captures extremely sparse network topologies, such as a ring or line topology.
	
	Furthermore, we conduct a comparison between the proposed algorithm and conventional centralized FL \cite{tran2019federated} in the decentralized network. In this scenario, a centrally located device coordinates the other devices, resulting in a communication structure resembling a star topology. It is important to note that centralized FL imposes strict chronological requirements, where the central device can only broadcast the model after aggregating the local models sequentially, i.e., in an uplink and downlink fashion. Consequently, the communication latency of centralized FL is twice that of MIMO OA-DFL, even for the same number of training rounds.
	We model the communication channels as independently and identically distributed (i.i.d.) Rayleigh fading, and the signal-to-noise ratio (SNR) at the transmitter side, defined as $P_0/\sigma_n^2$, is set equal for all devices. The learning configuration remains the same as the one described in Section \ref{sec-simluation}-A.
	
	To implement full-duplex over-the-air model aggregation, multiple antennas are necessary to provide sufficient degrees of freedom (DoF) for optimization. Therefore, we initially investigate the impact of the number of transmitter and receiver antennas, where the number of transmit (Tx) and receive (Rx) antennas are equal.    Unless otherwise specified, we adopt the following default settings: training round  $T=150$, the number of devices $M=30$, transmitter SNR $=20$ dB, maximum transmission power $P_0=1$ W, ${N_\text{T}}= {N_\text{R}}=20$, optimization-related parameters $J_{\max}=20$, $I_{1{\max}}=50$, $I_{2{\max}}=50$, $\lambda=0.02$, and $\omega=0.1$. The results are averaged over 30 Monte Carlo trials.
	\begin{figure}[htbp]
		\vspace{-1em}
		\centering
		\includegraphics[width=1\linewidth]{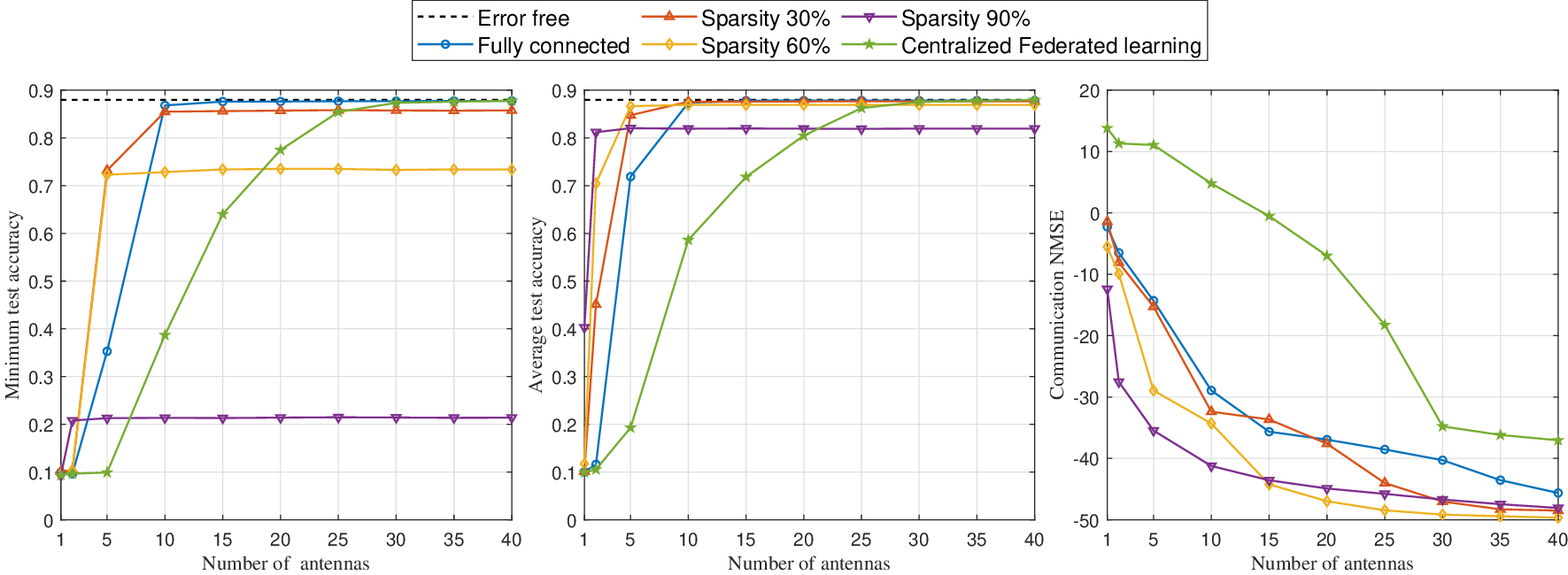}
			\vspace{-1em}
		\caption{Minimum and average test accuracy as well as communication NMSE (dB) versus the number of antennas in different topologies.}
		\label{diff_ntx}
		\vspace{-1em}
	\end{figure}

	In Fig.~\ref{diff_ntx}, we investigate the relationship between the Tx/Rx antenna size and three performance metrics: minimum test accuracy, average test accuracy, and communication normalized mean square error (NMSE). The communication NMSE is obtained by averaging the NMSE across different training rounds.
	We see from Fig.~\ref{diff_ntx} that sparse topology can achieve the highest training accuracy compared to dense topology when the number of antennas is relatively low. However, when the system has sufficient antennas, the performance of sparse topology is inferior to that of dense topology, which is particularly pronounced in terms of the minimum test accuracy (left subgraph).  For the communication NMSE, topologies with sparser structures exhibit less communication error, while denser topologies demonstrate higher error. Additionally, centralized FL suffers from significant errors due to its heavy reliance on the central device. A detailed analysis of the NMSE sheds light on the learning performance behaviors.
	
	Regarding test accuracy, high sparsity topologies (60\% and 90\%) achieve excellent average accuracy with minimal requirements. However, their performance in terms of minimum accuracy is poor, which is resulting from the limitation imposed by the mixing matrix where high ${\delta(\mathbf{W})}$ results in significant discrepancies in the local parameters from the global model average. On the other hand, non-sparse topologies (30\% and fully connected) exhibit a gradual increase in accuracy with the growth of the number of antennas. When the number of antennas is sufficient ($\geq 30$), the accuracies of fully connected topology and centralized FL are the same, achieving the performance consistent with the error free bound.
	Antenna requirements and topological sparsity represent a fundamental trade-off, and the optimal scenario involves achieving excellent performance with lower requirements, such as 30\% sparsity with 10 Tx/Rx antennas.
	
	We then evaluate the system performance versus transmitter-side SNR. As shown in Fig.~\ref{diff_SNR}, centralized FL continues to exhibit the highest error in terms of communication NMSE, and it struggles to perform model training effectively at low SNR regions (below $-10$ dB) due to substantial transmission errors.
	Notably, under the default settings, dense topologies exhibit better performance, consistently outperforming the sparse topologies across all transmitter-side SNR regions.
	\begin{figure}[h]
		\vspace{-1em}
		\centering
		\includegraphics[width=1\linewidth]{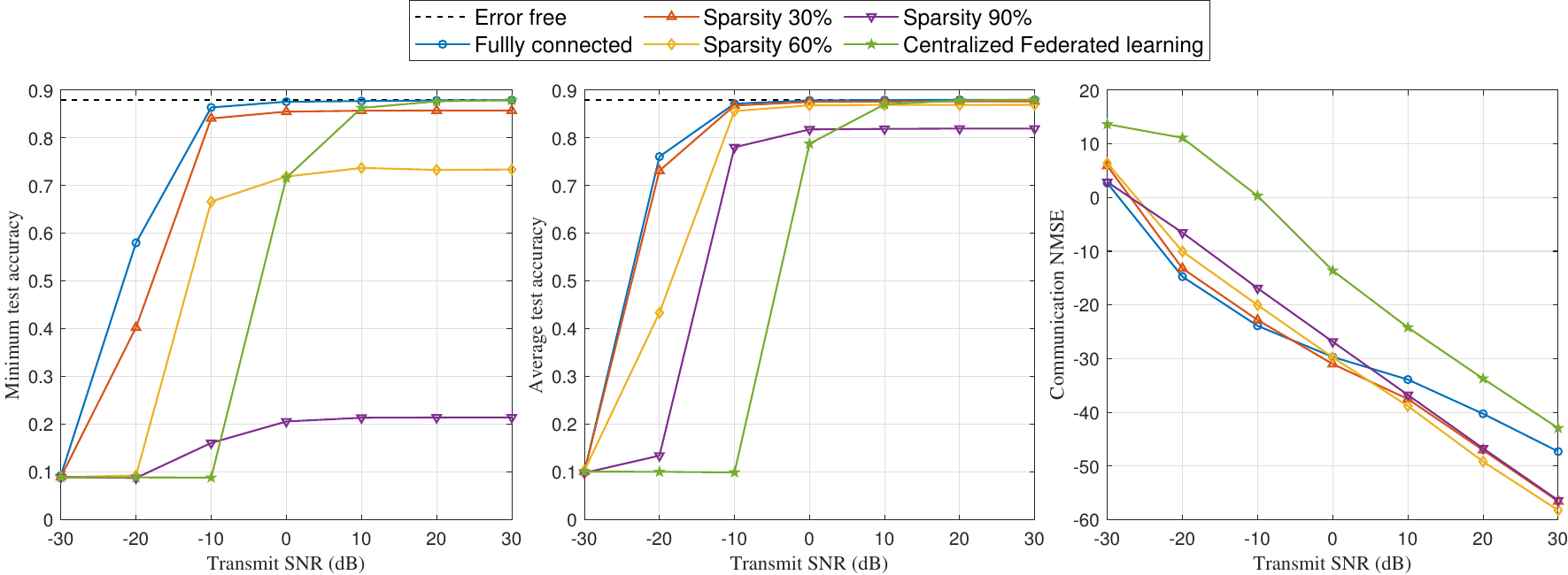}
			\vspace{-1em}
		\caption{Minimum and average test accuracy as well as communication NMSE (dB) versus transmitter SNR in different topologies.}
		\label{diff_SNR}
		\vspace{-1em}
	\end{figure}
	
	In Fig.~\ref{diff_User}, we investigate the impact of the number of devices on the system performance. 
	Due to limited communication resources, we see that the communication NMSE increases as the number of devices grows. The non-sparse topologies generally experience higher  NMSE compared to the sparse topologies.  In terms of test accuracy, with the exception of the fully connected topology, both minimum and average accuracy improve as the number of devices increases, particularly for sparser topologies. This behavior is attributed to that the second-largest squared eigenvalue decreases as the number of devices increases for a fixed sparsity level, while the impact of communication error remains insignificant within this range of device numbers. These findings highlight that in scenarios with a large number of devices, sparser topologies offer advantages due to lower communication requirements and, consequently, a better trade-off between communication and learning.
	\begin{figure}[htbp]
		\vspace{-1em}
		\centering
		\includegraphics[width=1\linewidth]{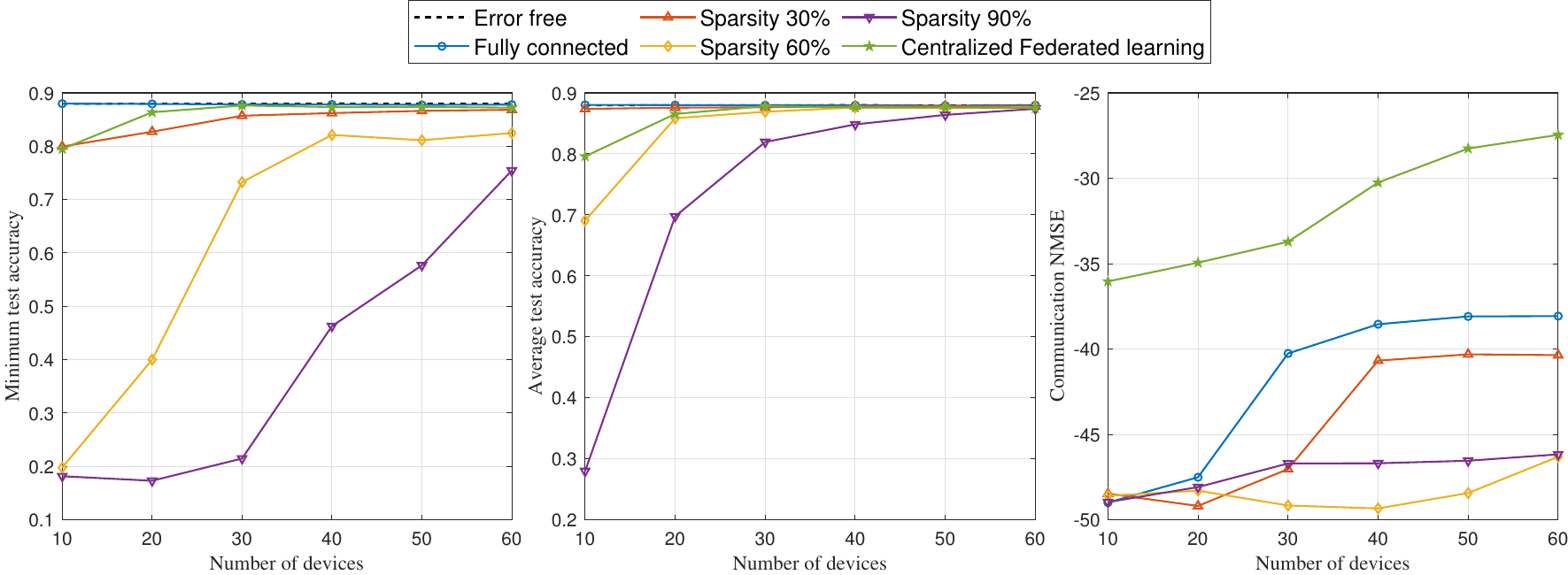}
			\vspace{-1em}
		\caption{Minimum and average test accuracy as well as communication NMSE (dB) versus number of devices in different topologies.}
		\label{diff_User}
			\vspace{-1em}
	\end{figure}
	\vspace{-1em}
	\subsection{Performance Comparison With Benchmarks}
	In this subsection, we present a comparison between the proposed algorithm and state-of-the-art schemes in terms of their performance under network topologies with sparsity levels of 30\% and 60\%. 
	The network topologies for  30\% and 60\% sparsity are shown in Fig.~\ref{0.3sp} and Fig.~\ref{0.6sp}, respectively. 

	\begin{figure}[htbp]
		\centering
		\begin{minipage}[t]{0.48\textwidth}
			\centering
			\includegraphics[width=6cm]{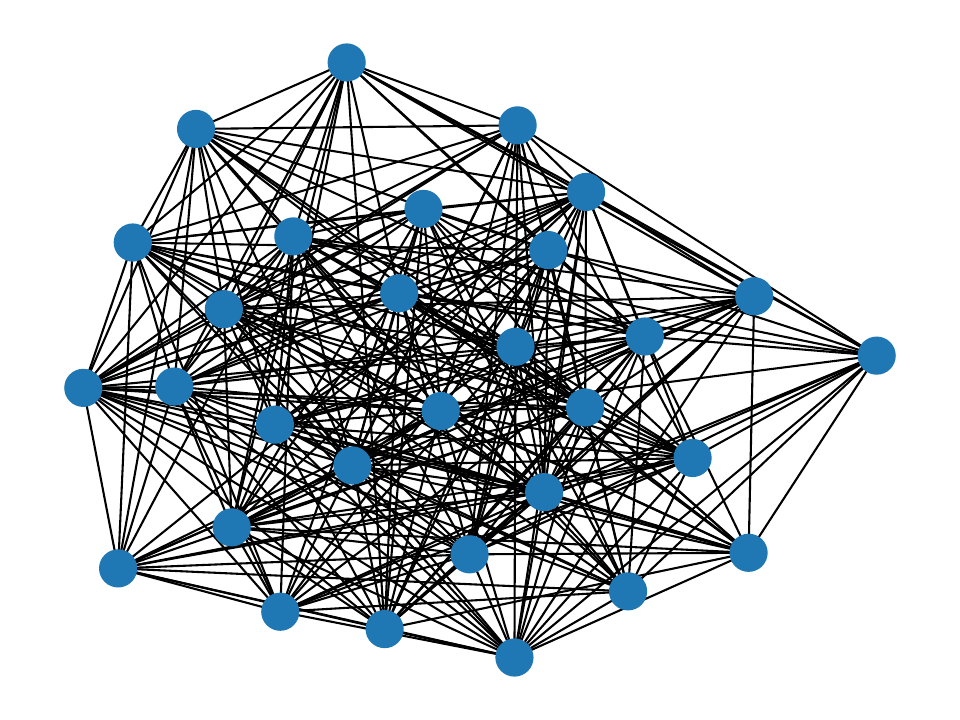}
			\caption{30\% Sparsity  topology}
			\label{0.3sp}
		\end{minipage}
		\begin{minipage}[t]{0.48\textwidth}
			\centering
			\includegraphics[width=6cm]{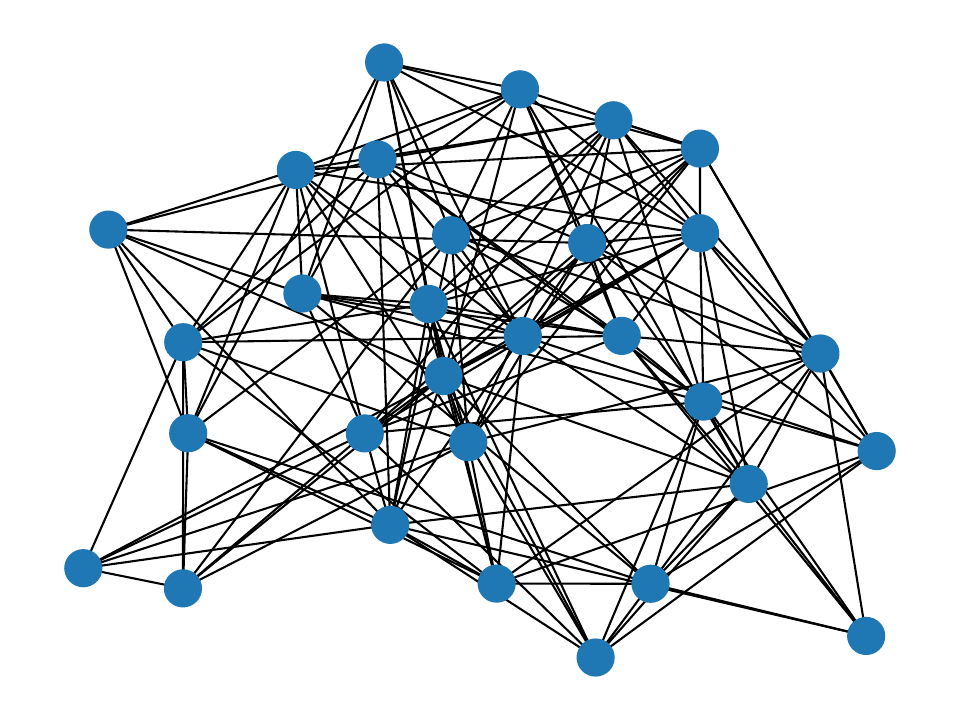}
			\caption{60\% Sparsity  topology}
			\label{0.6sp}
		\end{minipage}
	\vspace{-1em}
	\end{figure}
	
	The benchmarks to evaluate the performance of the proposed algorithm are as follows:
	\begin{itemize}
		\item \textbf{Joint optimization with separate over-the-air aggregation (JO with SOA)}: In this benchmark, each device sequentially acts as a central server to perform over-the-air aggregation in a time-division fashion during each training round. We jointly optimize the mixing matrix and beamformers, with the beamforming design being a special case of the over-the-air design presented in \cite{9955571}. It should be noted that the communication latency in this scheme is $M$-times larger than that of the proposed algorithm.
		\item \textbf{Digital communication without mixing matrix optimization (DC w/o MMO)}: In this benchmark, each model parameter is quantized to 16 bits and transmitted reliably with a channel capacity-achieving rate. During each training round, devices sequentially broadcast their model parameter to their neighbors, and a random mixing matrix is applied. The communication overhead in this scheme is significantly larger than that of our proposed algorithm due to the transmission protocol and capacity limitations.
		\item \textbf{Zero-forcing beamforming without mixing matrix optimization (ZFB w/o MMO)}:  In this benchmark, instead of minimizing the mean square error (MMSE), we optimize the transmit and receive beamforming using the zero-forcing criterion \cite{lin2022distributed}. The objective is to force the aggregated model parameter to approach the desired ground-truth value regardless of the channel noise. A random mixing matrix is applied in this scheme.
		\item  \textbf{MMSE beamforming without mixing matrix optimization (MB w/o MMO)}:
		In this benchmark, we optimize the beamforming vectors using our proposed algorithm with a given random mixing matrix.
		\item  \textbf{Error free communication with optimized mixing matrix (Error free case)}:  In this benchmark, we assume all communication channels are noiseless (i.e., $\sigma_n^2=0$). All devices exchange model parameters with perfect reliability and update their local model by $\hat{\xx}_{i}^{(t+\frac{1}{2})} = \xx_{i}^{(t + \frac{1}{2})},\forall i \in [M]$. We use the optimized mixing matrix (with the smallest possible second-largest squared eigenvalue). This scheme represents the optimal learning performance.
	\end{itemize}
	\begin{figure}[htbp]
			\vspace{-1em}
		\centering
		\begin{minipage}[t]{0.48\textwidth}
			\centering
			\includegraphics[width=8cm]{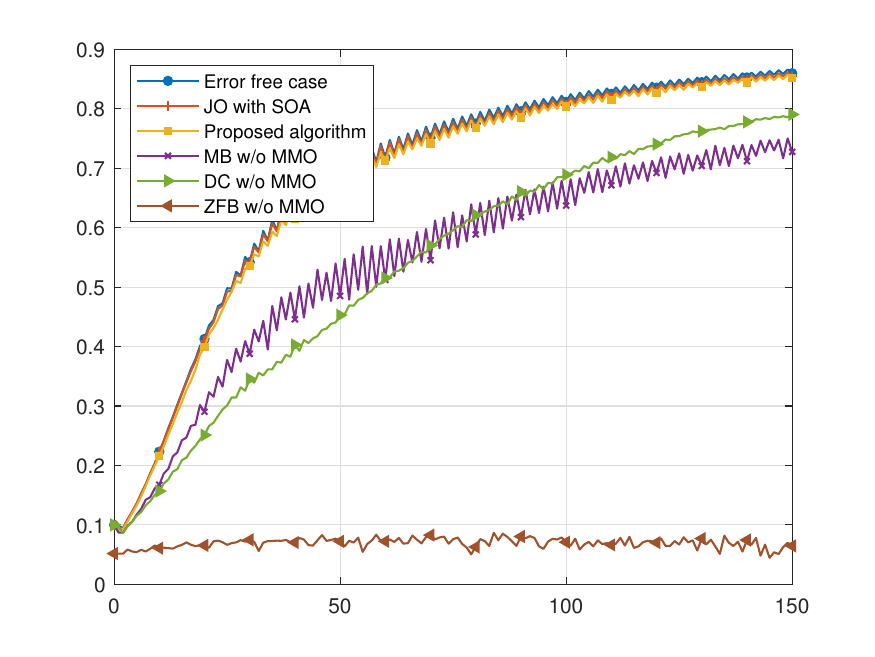}
		\end{minipage}
		\begin{minipage}[t]{0.48\textwidth}
			\centering
			\includegraphics[width=8cm]{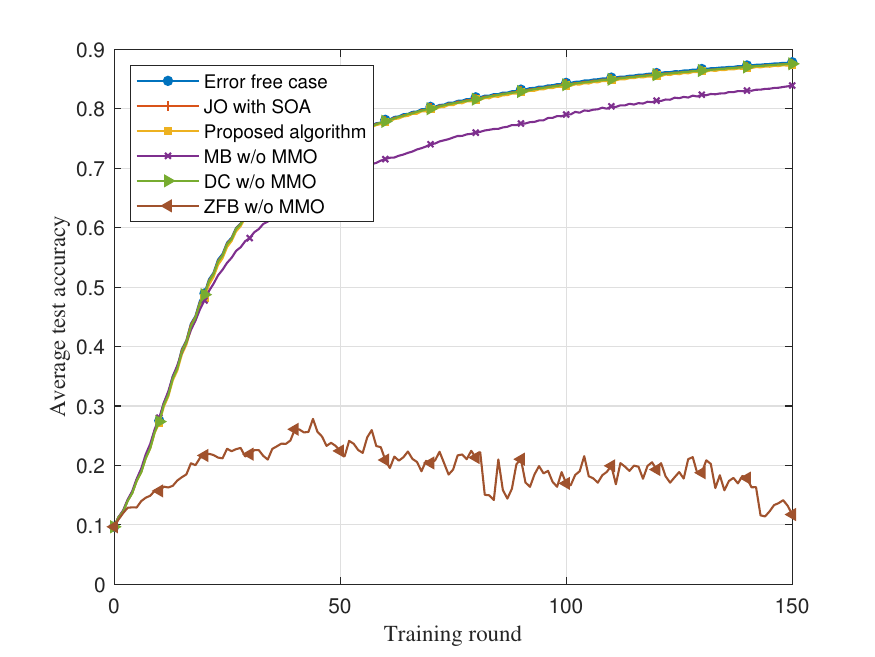}
		\end{minipage}
		\vspace{-1em}
		\caption{Minimum and average test accuracy versus training round under 30\% sparsity topology based on different schemes.}	
		\label{03sp_comparison}
		\centering
			\vspace{-1em}
	\end{figure}

	We conduct simulations with SNR set to 5dB, $M$ set to 30, and ${N_\text{T}}$ and ${N_\text{R}}$ set to 10 while keeping all other simulation setups the same as in Section \ref{sec-simluation}-B. The results are averaged over 30 Monte Carlo trials. We use the training round as the abscissa since different schemes require different communication times for one training procedure, where one training round corresponds to one DFL training process explained in Section II-A.
	
	In Fig.~\ref{03sp_comparison}, we compare the accuracy of the proposed algorithm with the benchmarks under 30\% sparsity topology. The results demonstrate that the proposed algorithm achieves nearly the same accuracy as the JO with SOA scheme while consuming $M=30$ times less communication time. Both of these schemes exhibit near-optimal performance as the error free case.
	Although DC w/o MMO scheme performs well in average accuracy, it lags in minimum accuracy due to the limited performance of the consensus, which is significantly impacted by the mixing matrix. Furthermore, the MB w/o MMO scheme suffers from both communication errors and significant discrepancies in model parameters, resulting in underperformance in both subgraphs. Moreover, ZFB w/o MMO scheme, without considering channel noise, does not perform well in this configuration.
	\begin{figure}[htbp]
			\vspace{-1em}
		\centering
		\begin{minipage}[t]{0.48\textwidth}
			\centering
			\includegraphics[width=8cm]{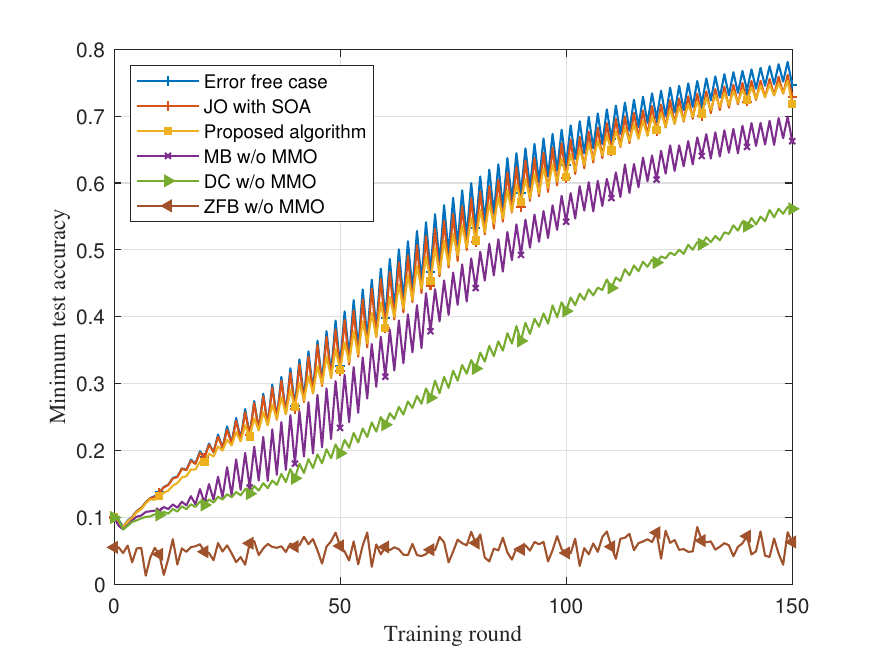}
		\end{minipage}
		\begin{minipage}[t]{0.48\textwidth}
			\centering
			\includegraphics[width=8cm]{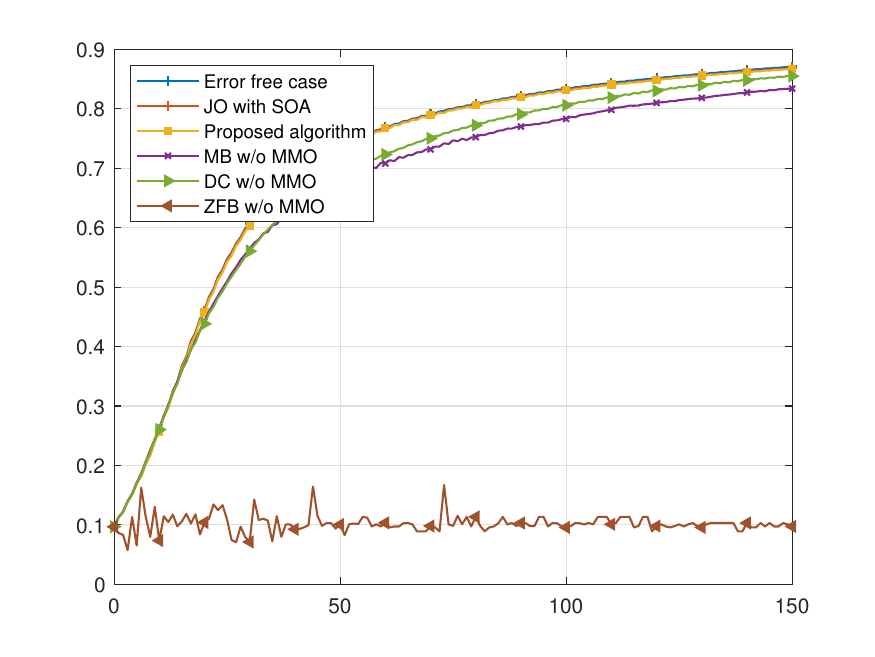}
		\end{minipage}
		\vspace{-1em}
		\caption{Minimum and average test accuracy versus training round under 60\% sparsity topology based on different schemes.}	
		\label{06sp_comparison}
		\centering
			\vspace{-1em}
	\end{figure}

	Fig.~\ref{06sp_comparison} provides a comparison of the results obtained under  60\% sparsity topology. Analyzing the minimum test accuracy (left subgraph), we observe a higher degree of fluctuation in the accuracy curve compared to the results under  30\% sparsity topology. This fluctuation can be attributed to the larger ${\delta(\mathbf{W})}$ of the mixing matrix under the 60\% sparsity topology.  The increased ${\delta(\mathbf{W})}$ leads to greater discrepancies among the local models and poses more challenges in achieving consensus.
	Furthermore, the proposed algorithm exhibits comparable performance to that of JO with SOA scheme and achieves near-optimal accuracy in both minimum and average accuracy cases. In contrast, the performance of other benchmarks significantly lags behind the proposed scheme due to the aforementioned reasons.
	\vspace{-1em}
	\section{Conclusions} \label{sec-conclusion}
	In this paper, we investigated the design of the  MIMO-OA DFL system over decentralized \emph{ad hoc} networks. We utilized a mixing matrix mechanism to promote consensus and leveraged wireless beamforming technique to improve communication quality. We derived a rigorous convergence bound in the MIMO-OA DFL scheme by capturing the impact of communication error on the decentralized learning performance. This provided a systematic attempt to characterize DFL performance considering both the learning and communication aspects. Based on this, we formulated a joint optimization problem with respect to transceiver beamformers and mixing matrix. We proposed a novel low-complexity AO algorithm to solve this problem. Finally, simulation results demonstrated the communication and learning trade-off in different topologies and verified the superiority of our proposed algorithm.
	
					\vspace{-1em}
	\appendices
	\section{Proof of Proposition~\ref{theorem1}}
	\label{Appendix_Prf_covergence1}
	To view the MIMO OA-DFL process from a global perspective, we represent the training steps in Section \ref{sec2a} using the matrix form. Denote the concatenation of the model parameter and the stochastic gradients of all devices in training round $t$ as
	\begin{align}
		{\mathbf{X}}^{(t)}\!\! \triangleq \!\!\left[{\xx}_{1}^{(t)},\dots, {\xx}_{M}^{(t)}~\right] \!\in\! \R^{D\times M},\partial F(\mX^{(t)},{\boldsymbol{\xi}}^{(t)}) \!\triangleq\!\!\left[\nabla F(\xx^{(t)}_1,\mathbf{\xi}^{(t)}_1),\dots,F(\xx^{(t)}_M,\mathbf{\xi}^{(t)}_M)\right] \!\!\in\! \R^{D\times M}.
	\end{align}
	We first consider the error free case training iteration, which  can be expressed as
	\begin{align}\label{errorfreecase}
		{\mX^{(t+1)}}^{\star}=\mX^{(t)}\mW-\lambda\partial F(\mX^{(t)},{\boldsymbol{\xi}}^{(t)}),
	\end{align}
	where ${\mX^{(t+1)}}^{\star}$ denotes the desired model parameter matrix at round $t+1$  and $\mW$ is the mixing matrix. Due to the communication error, the MIMO OA-DFL iteration is 
	\begin{align}\label{errorcase_glo}
		\mX^{(t+1)}=\hat{\mX}^{(t+\frac{1}{2})}-\lambda\partial F(\mX^{(t)},{\boldsymbol{\xi}}^{(t)}),
	\end{align}
	where $\hat{\mX}^{(t+\frac{1}{2})}\triangleq \left[{\hat{\xx}}_{1}^{(t+\frac{1}{2})},\dots, \hat{\xx}_{M}^{(t+\frac{1}{2})}~\right] \in \R^{D\times M}$ denotes the  practical received signal matrix. By comparing with equation \eqref{errorfreecase}, the MIMO OA-DFL iteration \eqref{errorcase_glo} can be rewritten as 
	\begin{align}
		\mX^{(t+1)}&=\mX^{(t)}\mW- \lambda\partial F(\mX^{(t)},{\boldsymbol{\xi}}^{(t)})-\mE^{(t)} ,
	\end{align}
	where $\mE^{(t)}\triangleq \mX^{(t)}\mW-\hat{\mX}^{(t+\frac{1}{2})} \in \R^{D\times M}$ represents the communication error. 
	
	Under Assumptions 1-3, with $\lambda \leq {1}/{\omega}$, 
	we have
	\begin{small} 
		\begin{align}
			&\E f\left(\frac{\mX^{(t+1)}\1}{M}\right)
			=\E f\left(\frac{\mX^{(t)}\mW\1}{M}-\lambda \frac{(\partial F(\mX^{(t)}, {\boldsymbol{\xi}}^{(t)})+\tilde{\mE}^{(t)})\1}{M}\right)\nonumber\\
			\overset{(a)}{\leq}&\E f\left(\frac{\mX^{(t)}\1}{M}\right)-\lambda \E \bigg\langle\nabla f\left(\frac{\mX^{(t)}\1}{M}\right),\frac{(\partial F(\mX^{(t)}, {\boldsymbol{\xi}}^{(t)})+\tilde{\mE}^{(t)})\1}{M}\bigg\rangle +\frac{\omega\lambda^2}{2}\E\norm{\frac{(\partial F(\mX^{(t)}, {\boldsymbol{\xi}}^{(t)})+\tilde{\mE}^{(t)})\1}{M}}^2\nonumber\\
			\overset{(b)}{=}& \E f\left(\frac{\mX^{(t)}\1}{M}\right)-\frac{\lambda}{2}\E\norm{\nabla f\left(\frac{\mX^{(t)}\1}{M}\right)}^2-\frac{\lambda}{2}\E\norm{\frac{\partial F(\mX^{(t)} {\boldsymbol{\xi}}^{(t)})\1}{M}+\frac{\tilde{\mE}^{(t)}\1}{M}}^2\nonumber\\&+\frac{\lambda}{2}\E\norm{\nabla f\left(\frac{\mX^{(t)}\1}{M}\right)-\frac{\partial F(\mX^{(t)}, {\boldsymbol{\xi}}^{(t)})\1}{M}-\frac{\tilde{\mE}^{(t)}\1}{M}}^2+\frac{\omega\lambda^2}{2}\E\norm{\frac{\partial F(\mX^{(t)}, {\boldsymbol{\xi}}^{(t)})\1}{M}+\frac{\tilde{\mE}^{(t)}\1}{M}}^2\nonumber\\
			\overset{(c)}{\leq}&f\left(\frac{\mX^{(t)}\1}{M}\right)-\frac{\lambda}{2}\E\norm{\nabla f\left(\frac{\mX^{(t)}\1}{M}\right)}^2+\lambda\E\norm{\nabla f\left(\frac{\mX^{(t)}\1}{M}\right)-\frac{\partial F(\mX^{(t)}, {\boldsymbol{\xi}}^{(t)})\1}{M}}^2+\lambda\E\norm{\frac{\tilde{\mE}^{(t)}\1}{M}}^2,\label{ini_iter}
		\end{align}
	\end{small}where $\tilde{\mE}^{(t)} \triangleq\mE^{(t)}/\lambda$, $(a)$ is based on the Assumption \ref{as1}-\ref{as2}, $(b)$ is because $2\langle \aa,\bb\rangle=\norm{\aa}^2+\norm{\bb}^2-\norm{\aa-\bb}^2$, and $(c)$ is from  $\lambda \leq {1}/{\omega}$ and $\norm{\sum_{i = 1}^n \aa_i}^2 \leq n \sum_{i = 1}^n \norm{\aa_i}^2$.  From \cite[Eq. (10)]{lian2017can}, the term $\small \E\norm{\nabla f\left(\frac{\mX^{(t)}\1}{M}\right)-\frac{\partial F(\mX^{(t)}, {\boldsymbol{\xi}}^{(t)})\1}{M}}^2$ can be bounded as
	\begin{small}
		\begin{align}
			\E\norm{\nabla f\left(\frac{\mX^{(t)}\1}{M}\right)-\frac{\partial F(\mX^{(t)}, {\boldsymbol{\xi}}^{(t)})\1}{M}}^2	{\leq}\frac{\omega^2}{M}\sum_{i=1}^{M}\E\norm{\frac{\mX^{(t)}\1}{M}-\mX^{(t)}\ee_i}^2+\frac{\alpha^2}{M}.\label{24}
		\end{align}
	\end{small}We define  $\frac{1}{M}\sum_{i=1}^{M}{\E\norm{\frac{\mX^{(t)}\1}{M}-\mX^{(t)}\ee_i}^2}$ as the agreement error in round $t$, which is the main obstacle in the decentralized convergence analysis. We start by bounding $\small \Xi^{(t)}_i\triangleq{\E\norm{\frac{\mX^{(t)}\1}{M}-\mX^{(t)}\ee_i}^2}$:
	\vspace{-1em}
	\begin{small}
		\begin{align}	
			\Xi^{(t)}_i=&
			\E\norm{\frac{\mX^{(t-1)}{\mW}\1-\lambda\left(\partial F(\mX^{(t-1)}, {\boldsymbol{\xi}}^{(t-1)})+\tilde{\mE}^{(t)}\right)\1}{M}-(\mX^{(t-1)}{\mW}{\ee_i}-\lambda(\partial F(\mX^{(t-1)}, {\boldsymbol{\xi}}^{(t-1)})+\tilde{\mE}^{(t)}){\ee_i})}^2,\notag\\
			{=}&\lambda^2\E\norm{\sum_{j=0}^{t-1}\left(\partial F({\mX^{(j)}}, {\boldsymbol{\xi}}^{(j)})-\partial f({\mX^{(j)}})+\partial f({\mX^{(j)}})+\tilde{\mE}^{(j)}\right)\left(\frac{\1}{M}-{\mW^{t-j-1}}{\ee_i}\right)}^2,\notag\\	
			\leq&3\lambda^2{\E\norm{\sum_{j=0}^{t-1}\left(\partial F({\mX^{(j)}}, {\boldsymbol{\xi}}^{(j)})-\partial f({\mX^{(j)}})\right)\left(\frac{\1}{M}-{\mW^{t-j-1}}{\ee_i}\right)}^2},\notag\\+&3\lambda^2{\E\norm{\sum_{j=0}^{t-1}\partial f({\mX^{(j)}})\left(\frac{\1}{M}-{\mW^{t-j-1}}{\ee_i}\right)}^2}+3\lambda^2{\E\norm{\sum_{j=0}^{t-1}\tilde{\mE}^{(j)}\left(\frac{\1}{M}-{\mW^{t-j-1}}{\ee_i}\right)}^2}, \label{Qk_bound1}
		\end{align}
	\end{small}where we simp{\tiny }lify the derivation by assuming  $\mX^{(0)}=0$. For the  first  term on the RHS of inequality \eqref{Qk_bound1}, we have
	\begin{small}
		\begin{align}
			&\E\norm{\sum_{j=0}^{t-1}\left(\partial F({\mX^{(j)}} {\boldsymbol{\xi}}^{(j)})-\partial f({\mX^{(j)}})\right)\left(\frac{\1}{M}-{\mW^{t-j-1}}{\ee_i}\right)}^2,\nonumber\\
			\leq&\sum_{j=0}^{t-1}\E\norm{\left(\partial F({\mX^{(j)}}, {\boldsymbol{\xi}}^{(j)})-\partial f({\mX^{(j)}})\right)}^2_F\norm{\left(\frac{\1}{M}-{\mW^{t-j-1}}{\ee_i}\right)}^2
			{\leq}\frac{M\alpha^2}{1-{\delta(\mathbf{W})}},\label{termT1}
		\end{align}
	\end{small}where the last inequality is due to Assumption 3.  By following the analysis in \cite{lian2017can},  the second  term on the RHS of \eqref{Qk_bound1} can be bounded as
	\begin{small}
		\begin{align}
			&\E\norm{\sum_{j=0}^{t-1}\partial f({\mX^{(j)}})\left(\frac{\1}{M}-{\mW^{t-j-1}}{\ee_i}\right)}^2\leq3\sum_{j=0}^{t-1}\sum_{h=1}^{M}\E \omega^2	{\Xi^{(j)}_h}\norm{\left(\frac{\1}{M}-{\mW^{t-j-1}}{\ee_i}\right)}^2+3\sum_{j=0}^{t-1}\E \norm{\nabla f\left(\frac{{\mX^{(j)}}\1}{M}\right)\1^\top}^2\nonumber\\
			&\norm{\left(\frac{\1}{M}-{\mW^{t-j-1}}{\ee_i}\right)}^2
			\!\!+\!6\sum_{j=0}^{t-1}\left(\sum_{h=1}^{M}\E \omega^2	{\Xi^{(j)}_h}\!\!+\!\!\E \norm{\nabla f\left(\frac{{\mX^{(j)}}\1}{M}\right)\1^\top}^2\right)\frac{\sqrt{{\delta(\mathbf{W})}}^{k-j-1}}{1-\sqrt{{\delta(\mathbf{W})}}}\!\!+\!\!\frac{9n\beta^2}{(1-\sqrt{{\delta(\mathbf{W})}})^2}.\label{termT2}
		\end{align}
	\end{small}
	
	We then bound the last term on the RHS of \eqref{Qk_bound1}:
	\begin{small}
		\begin{align}
			&\E\norm{\sum_{j=0}^{t-1}\tilde{\mE}^{(j)}\left(\frac{\1}{M}-{\mW^{t-j-1}}{\ee_i}\right)}^2\nonumber\\
			=&\sum_{j=0}^{t-1}\E\norm{\tilde{\mE}^{(j)}\left(\frac{\1}{M}-{\mW^{t-j-1}}{\ee_i}\right)}^2+\sum_{j\neq j'}^{k-1}\E\bigg\langle \tilde{\mE}^{(j)}\left(\frac{\1}{M}-{\mW^{t-j-1}}{\ee_i}\right),\tilde{\mE}^{(j')}\left(\frac{\1}{M}-\mW^{t-j'-1}{\ee_i}\right)\bigg\rangle,\nonumber\\
			\leq&\sum_{j=0}^{t-1}\E\norm{\tilde{\mE}^{(j)}}^2\norm{\frac{\1}{M}-{\mW^{t-j-1}}{\ee_i}}^2+\sum_{j\neq j'}^{k-1}\E\norm{\tilde{\mE}^{(j)}}\norm{\frac{\1}{M}-{\mW^{t-j-1}}{\ee_i}}\norm{\tilde{\mE}^{(j')}}\norm{\frac{\1}{M}-\mW^{t-j'-1}{\ee_i}},\nonumber\\
			\leq&\sum_{j=0}^{t-1}\E\norm{\tilde{\mE}^{(j)}}^2{\delta(\mathbf{W})}^{k-j-1}+\sum_{j\neq j'}^{k-1}\E\left(\frac{\norm{\tilde{\mE}^{(j)}}^2}{2}+\frac{\norm{\tilde{\mE}^{(j')}}^2}{2}\right)\norm{\frac{\1}{M}-{\mW^{t-j-1}}{\ee_i}}\norm{\frac{\1}{M}-\mW^{t-j'-1}{\ee_i}},\nonumber\\
			\overset{(a)}{\leq}&\sum_{j=0}^{t-1}\E\norm{\tilde{\mE}^{(j)}}^2{\delta(\mathbf{W})}^{k-j-1}+\sum_{j\neq j'}^{k-1}\E\left(\frac{\norm{\tilde{\mE}^{(j)}}^2}{2}\!\!+\!\!\frac{\norm{\tilde{\mE}^{(j')}}^2}{2}\right){\delta(\mathbf{W})}^{k-\frac{j+j'}{2}-1},\nonumber\\
			\leq&\sum_{j=0}^{t-1}\E\norm{\tilde{\mE}^{(j)}}^2\!\!\!\!{\delta(\mathbf{W})}^{k-j-1}\!\!\!+\!\!\sum_{j\neq j'}^{k-1}\E\norm{\tilde{\mE}^{(j)}}^2\!\!\!{\delta(\mathbf{W})}^{k-\frac{j+j'}{2}-1} \label{termT3}
			\!\!\!\!\leq\!\sum_{j=0}^{t-1}\E\norm{\tilde{\mE}^{(j)}}^2\left({\delta(\mathbf{W})}^{k-j-1}\!\!+\!\!\frac{2\sqrt{{\delta(\mathbf{W})}}^{k-j-1}}{1-\sqrt{{\delta(\mathbf{W})}}}\!\!\right)\!\!,
		\end{align}
	\end{small}where $(a)$ follows from {Lemma~\ref{lemma_Wmatrix}}. Plugging \eqref{termT1}, \eqref{termT2} and \eqref{termT3} back to \eqref{Qk_bound1}, we obtain the bound for $\Xi^{(t)}_i$:
	\begin{small}
		\begin{align}
			\Xi^{(t)}_i
			\!\!\leq&9\lambda^2\sum_{j=0}^{t-1}\E \norm{\nabla f\left(\frac{{\mX^{(j)}}\1}{M}\right)\1^\top}^2\left({\delta(\mathbf{W})}^{k-j-1}+\frac{2\sqrt{{\delta(\mathbf{W})}}^{k-j-1}}{1-\sqrt{{\delta(\mathbf{W})}}}\right)+\!\!9\lambda^2\sum_{j=0}^{t-1}\sum_{h=1}^{M}\E \omega^2	{\Xi^{(j)}_h}\left({\delta(\mathbf{W})}^{k-j-1}\!\!+\!\!\frac{2\sqrt{{\delta(\mathbf{W})}}^{k-j-1}}{1-\sqrt{{\delta(\mathbf{W})}}}\right)\nonumber\\
			&+3\lambda^2\sum_{j=0}^{t-1}\E\norm{\tilde{\mE}^{(j)}}^2\left({\delta(\mathbf{W})}^{k-j-1}+\frac{2\sqrt{{\delta(\mathbf{W})}}^{k-j-1}}{1-\sqrt{{\delta(\mathbf{W})}}}\right)+\frac{3\lambda^2n\alpha^2}{1-{\delta(\mathbf{W})}}\!\!+\!\!\frac{27\lambda^2n\beta^2}{(1-\sqrt{{\delta(\mathbf{W})}})^2}.
		\end{align}
	\end{small}Therefore, we have the following bound:
	\begin{small}
		\begin{align}
			\frac{1}{M}\sum_{i=1}^{M}&{\E\norm{\frac{\mX^{(t)}\1}{M}-\mX^{(t)}\ee_i}^2}
			\leq\frac{3\lambda^2M\alpha^2}{1-{\delta(\mathbf{W})}}\!\!+\!\!\frac{27\lambda^2M\beta^2}{(1-\sqrt{{\delta(\mathbf{W})}})^2}\!\!+\!\!9\lambda^2\sum_{j=0}^{t-1}\E \norm{\nabla f\left(\frac{{\mX^{(j)}}\1}{M}\right)\1^\top}^2\left({\delta(\mathbf{W})}^{k-j-1}\!\!+\!\!\frac{2\sqrt{{\delta(\mathbf{W})}}^{k-j-1}}{1-\sqrt{{\delta(\mathbf{W})}}}\right)\notag\\
			+9\lambda^2\omega^2&\sum_{j=0}^{t-1} \sum_{i=1}^{M}{\E\norm{\frac{\mX^{(j)}\1}{M}-\mX^{(j)}\ee_i}^2}\left({\delta(\mathbf{W})}^{k-j-1}\!\!+\!\!\frac{2\sqrt{{\delta(\mathbf{W})}}^{k-j-1}}{1-\sqrt{{\delta(\mathbf{W})}}}\right)\!\!+\!\!3\lambda^2\sum_{j=0}^{t-1}\E\norm{\tilde{\mE}^{(j)}}^2\left({\delta(\mathbf{W})}^{k-j-1}\!\!\!+\!\!\frac{2\sqrt{{\delta(\mathbf{W})}}^{k-j-1}}{1-\sqrt{{\delta(\mathbf{W})}}}\right)\!.\label{consen_bound}
		\end{align}
	\end{small}
	
	Note that the agreement error $\frac{1}{M}\sum_{i=1}^{M}{\E\norm{\frac{\mX^{(t)}\1}{M}-\mX^{(t)}\ee_i}^2}$ appears on both sides of the inequality. Summing \eqref{consen_bound} from $t=0$ to $T-1$, by rearranging the summation and relaxing the inequality,  we obtain the final bound of the agreement error as
	\begin{small}
		\begin{align}
			&\frac{1}{M}	\sum_{t=0}^{T-1} \sum_{i=1}^{M}{\E\norm{\frac{\mX^{(t)}\1}{M}-\mX^{(t)}\ee_i}^2}\!\!\leq\!\!\frac{3\lambda^2M\alpha^2}{(1-{\delta(\mathbf{W})})\left(1-\frac{27}{(1-\sqrt{{\delta(\mathbf{W})}})^2}M\lambda^2\omega^2\right)}T\!\!+\!\!\frac{27\lambda^2M\beta^2}{(1-\sqrt{{\delta(\mathbf{W})}})^2\left(1-\frac{27}{(1-\sqrt{{\delta(\mathbf{W})}})^2}M\lambda^2\omega^2\right)}T\notag\\&\!\!+\!\!\frac{27\lambda^2}{(1-\sqrt{{\delta(\mathbf{W})}})^2\left(1-\frac{27}{(1-\sqrt{{\delta(\mathbf{W})}})^2}M\lambda^2\omega^2\right)}\!\!\sum_{t=0}^{T-1}\!\E\! \norm{\nabla f\left(\frac{\mX^{(t)}\1}{M}\right)\1^\top}^2\!\!\!+\!\!\frac{9\lambda^2}{(1-\sqrt{{\delta(\mathbf{W})}})^2\left(1-\frac{27}{(1-\sqrt{{\delta(\mathbf{W})}})^2}M\lambda^2\omega^2\right)}\!\!\sum_{t=0}^{T-1}\!\!\E\!\norm{\tilde{\mE}^{(t)}}^2\!\!.\label{sum_iter_con}
					\vspace{-3em}
		\end{align}
	\end{small}

	Finally, we sum the inequality \eqref{ini_iter} from $t=0$ to $T-1$ while using \eqref{24} and \eqref{sum_iter_con}, which completes the proof of {Proposition~\ref{theorem1}}.
		\vspace{-1em}
	\section{Proof of Proposition~\ref{Corollary:MSE}}\label{app_c}
	The term $\E\norm{{\mE}^{(t)}}^2_F$  is given by
	\begin{align}
		\E\norm{{\mE}^{(t)}}^2_F=\E\sum_{i=1}^{M}\sum_{d=1}^{D}\left[\left|x_{i}^{(t + \frac{1}{2})}[d]-{\hat x}_{i}^{(t + \frac{1}{2})}[d]\right|^2\right]\label{error_term1}
	\end{align}
	
	By substituting \eqref{free_aggre}, \eqref{nomalize}, \eqref{complex1}, \eqref{wholechannel} and  \eqref{rece1} into \eqref{error_term1}, we obtain
	\begin{align}
		\E\norm{\mE^{(t)}}^2_F&\!=\!\sum_{i=1}^{M}\sum_{c=1}^{C}\E\left[\left|\sum_{j\in \cM_i}(w_{ij}r^{(t)}_j[c]v^{(t)}_j-r^{(t)}_j[c]({\ff^{(t)}_{i}})^{\mathrm{H}}\mathbf{H}_{\langle i,j\rangle}^{(t)}\uu^{(t)}_j)\right|^2\!\!+\!\!\left|({\ff^{(t)}_{i}})^{\mathrm{H}}\mathbf{n}_{k,i}[c]\right|^2\right]\label{error1}
	\end{align}
	where the first term on the RHS  represents the misalignment error, and the other represents the error due to channel noise. Similarly, the  term $\E\norm{{\mE^{(t)}\1}}^2$ can be expressed as
	\begin{align}
		\E\norm{{\mE^{(t)}\1}}^2&\!\!=\!\!\sum_{c=1}^{C}\E\!\!\left[\left|\sum_{i=1}^{M}\sum_{j\in\cM_i}\left(w_{ij}r^{(t)}_j[c]v^{(t)}_j\!\!-\!r^{(t)}_j[c]({\ff^{(t)}_{i}})^{\mathrm{H}}\mathbf{H}_{\langle i,j\rangle}^{(t)}\uu^{(t)}_j\right)\right|\!\!+\!\!\left|\sum_{i=1}^{M}({\ff^{(t)}_{i}})^{\mathrm{H}}\nn_{k,i}[c]\right|^2\!\right]\label{error2}
	\end{align}
	Based on the correlation assumption in \eqref{corr}, we have $\E[r^{(t)}_i[c]^{\ast}r^{(t)}_i[c]]\!=\!2,\forall i,\forall c$ and $\E[r^{(t)}_i[c_1]^{\ast}r^{(t)}_j[c_2]]\!\!=\!0,\!\forall (i \neq j) \!\cup\! \forall (c_1 \neq c_2)\!$. \!By using them, we expand \eqref{error1} and \eqref{error2} and finally obtain {Proposition \ref{Corollary:MSE}}.
		\vspace{-1.5em}
	\section{Proof of Proposition~\ref{sub_problem21}}\label{app_sub_problem21}
	In {Proposition \ref{Corollary:MSE}},  $\E\norm{\mE}^2_F$ is clearly a convex function with respect to $\mW$. Besides, $\E\norm{{\mE\1}}^2$ is also a convex function of $\mW$ by noting
		\vspace{-0.5em}
	\begin{align}
		\sum_{p=1}^{M}\sum_{i,j\in \cM_p}^{}w_{ip}w_{jp}{v_{p}}^2=\sum_{p=1}^{M}v_{p}^2\bigg(\sum_{i\in \cM_p}^{n}w_{ip}\bigg)^2.
	\end{align}
	Therefore, the objective function in P7 is convex. Additionally, constraint \eqref{sub21_cons} is an affine constraint with respect to $\mW$. Then we only need to prove the convexity of the eigenvalue-related constraint ${\delta(\mathbf{W})}\leq \hat{\delta}$.  
	
	For  a symmetric matrix $\mX \in \R^{n\times n}$, using the variational characterization, the sum of the $k$ largest squared eigenvalues can be expressed as
				\vspace{-0.5em}
	\begin{align}
		\sum_{i=1}^{k}\lambda_i(\mX^2)=&\!\!\sup_{\vv_1,\dots\vv_k}\!\!\Big\{\sum_{i=1}^{k}{\vv_i}^T\mX^2\vv_i,\bigg|{\vv_i}^\mathrm{T}{\vv_j}=\text{\slcurly{$1,~i= j$ \notag\\ $0,~i\neq j$} }
		\Big\}\\=&\sup_{}\left\{\Tr\left({\mV^T\mX^2\mV}\right)|\mV^T\mV=I\right\}\notag\\
		=&\sup\left\{\Tr\left({(\mX \mV)^T(\mX \mV)}\right)|\mV^T\mV=I\right\}\notag\\
		=&\sup\left\{{\norm{\mX \mV}}_F^2|\mV^T\mV=I\right\},\label{sup_V}
	\end{align}
	where $\mV\triangleq[\vv_1,\vv_2,\dots\vv_k] \in\R^{n\times k}$. Note that \eqref{sup_V} is a  point-wise supremum of convex functions; hence, it is a convex function of ${\mX}$\cite{boyd_vandenberghe_2004}. 
	
	Note that we in fact have ${{\delta(\mathbf{W})}}=\sum_{i=1}^{2}\lambda_i(\mW^2)-1$ for symmetric doubly stochastic matrix $\mW$. It then follows that the constraint ${\delta(\mathbf{W})} \leq \hat{\delta}$ is convex. Hence, the problem P7 is convex.
	
	\bibliographystyle{IEEEtran}
	\bibliography{IEEEabrv,Dec-fl}
\end{document}